\documentclass{fundam}
\usepackage{latexsym,amssymb,amsmath,upgreek}
\usepackage{xspace}
\usepackage{enumerate}
\usepackage{mathrsfs}
\usepackage{setspace}
\usepackage{color}

\long\def\COMMENT#1{}

\renewcommand{\And}{\wedge}
\newcommand{\Or}{\vee}

\newcommand{\syntree}[1]{\mathit{S}_{#1}}

\newcommand{\false}{\mathbf{false}}
\newcommand{\true}{\mathbf{true}}

\newcommand{\comment}[1]{}

\newcommand{\defAs}
   {\mbox{$\:= \! \! \raisebox{-0.5 ex}[0 ex][0 ex]{\tiny Def}\:$}}

\newcommand{\Hf}{H_{\varphi}}

\newcommand{\model}{\ensuremath{\mbox{\boldmath $\mathcal{M}$}}\xspace}

\newcommand{\pow}{\mbox{\rm pow}}
\newcommand{\Rel}{\mbox{\rm Rel}}

\newcommand{\SubF}{\mathit{SubF}}
\newcommand{\Kr}{\mathit{K}}
\newcommand{\K}{\mathsf{K}}
\newcommand{\T}{\mathsf{T}}
\newcommand{\Ac}{\mathbf{5}}
\newcommand{\B}{\mathsf{B}}
\newcommand{\Aq}{\mathbf{4}}
\newcommand{\D}{\mathsf{D}}
\newcommand{\Sc}{\mathsf{S5}}
\newcommand{\Kqc}{\mathsf{K45}}

\newcommand{\TKqc}{\tau_{\Kqc}}
\newcommand{\TKqcuno}{\tau_{\Kqc}^{1}}
\newcommand{\TKqcdue}{\tau_{\Kqc}^{2}}

\newcommand{\QLQSR}{\ensuremath{\mbox{$4\mathit{LQS}^{R}$}}\xspace}
\newcommand{\TLQSR}{\ensuremath{\mbox{$3\mathit{LQS}^{R}$}}\xspace}
\newcommand{\QLQS}{\ensuremath{\mbox{$4\mathit{LQS}$}}\xspace}

\begin{document}

\title{On the satisfiability problem for a 4-level quantified syllogistic and some applications to modal logic (extended version)}

\author{Domenico Cantone \\
Dipartimento di Matematica e Informatica, Universit\`a di
Catania\\
      Viale A. Doria 6, I-95125 Catania, Italy \\
cantone@dmi.unict.it\\
\and Marianna Nicolosi Asmundo \\
Dipartimento di Matematica e Informatica, Universit\`a di
Catania\\
      Viale A. Doria 6, I-95125 Catania, Italy \\
      nicolosi@dmi.unict.it}
\maketitle


\begin{abstract}
We introduce a multi-sorted stratified syllogistic, called $\QLQSR$,
admitting variables of four sorts and a restricted form of
quantification over variables of the first three sorts, and prove that
it has a solvable satisfiability problem by showing that it enjoys a
small model property.  Then, we consider the fragments $(\QLQSR)^h$ of
$\QLQSR$, consisting of $\QLQSR$-formulae whose quantifier prefixes
have length bounded by $h \geq 2$ and satisfying certain syntactic
constraints, and prove that each of them has an \textsf{NP}-complete
satisfiability problem.  Finally we show that the modal logic $\Kqc$
can be expressed in $(\QLQSR)^3$.
\end{abstract}

\section{Introduction}
Most of the decidability results in computable set theory concern
one-sorted multi-level syllogistics, namely collections of formulae
admitting variables of one sort only, which range over the von Neumann
universe of sets (see \cite{CFO89,COO01} for a thorough account of the
state-of-art until 2001).
Only a few stratified syllogistics, where
variables of different sorts are allowed, have been investigated,
despite the fact that in many fields of computer science and
mathematics often one has to deal with multi-sorted languages.\footnote{The locutions `multi-level syllogistic' and `stratified syllogistic' were chosen by Jack Schwartz to
name many decidable fragments of computable set theory because he saw them as generalizations of Aristotelian syllogistics.}
For instance, in modal logics, one has to consider entities of different
types, namely worlds, formulae, and accessibility relations.

In \cite{FerOm1978} an efficient decision procedure was presented for
the satisfiability of the Two-Level Syllogistic language ($2LS$).
$2LS$ has variables of two sorts and admits propositional connectives
together with the basic set-theoretic operators $\cup,\cap,
\setminus$, and the predicate symbols $=,\in$, and $\subseteq$.  Then,
in \cite{CanCut90}, it was shown that the extension of $2LS$ with the
singleton operator and the Cartesian product operator is decidable.
Tarski's and Presburger's arithmetics extended with sets have been
analyzed in \cite{CCS90}.  Subsequently, in \cite{CanCut93}, a
three-sorted language $3LSSPU$ (Three-Level Syllogistic with
Singleton, Powerset and general Union) has been proved decidable.
Recently, in \cite{CanNic08}, it was shown that the language $\TLQSR$
(Three-Level Quantified Syllogistic with Restricted quantifiers) has a
decidable satisfiability problem.  $\TLQSR$ admits variables of three
sorts and a restricted form of quantification.  Its vocabulary
contains only the predicate symbols $=$ and $\in$.  In spite of that,
$\TLQSR$ allows one to express several constructs of set theory.
Among them, the most comprehensive one is the set-formation operator, which in
turn enables one to express other operators like the powerset
operator, the singleton operator, and so on.  In \cite{CanNic08} it is
also shown that the modal logic $\Sc$ can be expressed in a
fragment of $\TLQSR$, whose satisfiability problem is
\textsf{NP}-complete.

In this paper we present a decidability result for the satisfiability
problem of the set-theoretic language $\QLQSR$ (Four-Level Quantified
Syllogistic with Restricted quantifiers).  $\QLQSR$ is an extension of
$\TLQSR$ admitting variables of four sorts and a restricted form of
quantification over variables of the first three sorts.  In addition
to the predicate symbols $=$ and $\in$, its vocabulary contains also the
pairing operator $\langle \cdot,\cdot\rangle$.

We will prove that the theory $\QLQSR$ enjoys a small model property
by showing how one can extract, out of a given model satisfying a
$\QLQSR$-formula $\psi$, another model of $\psi$ but of bounded finite
cardinality.  The construction of the finite model extends the
decision algorithm described in \cite{CanNic08}.  Concerning complexity
issues, we will show that the satisfiability problem for each of the
fragments $(\QLQSR)^h$ of $\QLQSR$, whose formulae are restricted to
have their quantifier prefixes of length at most $h \geq 2$ and must
satisfy certain additional syntactic constraints to be seen later, is
\textsf{NP}-complete.

In addition to the modal logic $\Sc$, already expressible in the
language $\TLQSR$, it turns out that in $\QLQSR$ one can also
formalize several properties of binary relations (needed to define
accessibility relations of well-known modal logics) and some Boolean
operations over relations and the inverse operation over binary
relations.  We will also show that the modal logic $\Kqc$ can be
formalized in the fragment $(\QLQSR)^3$.  As is well-known, the
satisfiability problem for $\Kqc$ is \textsf{NP}-complete; thus our
alternative decision procedure for $\Kqc$ can be considered optimal in
terms of its computational complexity.

\section{The language $\QLQSR$}\label{language}
Before defining the language $\QLQSR$ of our interest, it is
convenient to present the syntax and the semantics of a more general,
unrestricted four-level quantified fragment, denoted $\QLQS$.
Subsequently, we will introduce suitable restrictions over the
formulae of $\QLQS$ to characterize the sublanguage $\QLQSR$.

\subsection{The unrestricted language $\QLQS$}\label{genericlanguage}

\paragraph{Syntax of $\QLQS$.}
The four-level quantified language $\QLQS$ involves the four collections
${\cal V}_{0}$, ${\cal V}_{1}$, ${\cal V}_{2}$, and ${\cal V}_{3}$ of variables. Each ${\cal V}_{i}$
contains variables of \emph{sort i}, denoted by $X^{i},Y^{i},Z^{i},\ldots$.
When we refer to variables of sort 0 we prefer to write $x,y,z,\ldots$ instead of $X^{0},Y^{0},Z^{0},\ldots$.
In addition to the variables in ${\cal V}_{2}$, terms of sort 2
include also \emph{pair terms} of the form $\langle x,y\rangle$, for
$x,y \in {\cal V}_{0}$.\\[.2cm]
$\QLQS$ \emph{quantifier-free atomic formulae} are classified as:
\begin{description}
\item[level $0$:] $x = y$, $x \in X^{1}$, for $x,y \in {\cal V}_{0},
X^{1} \in {\cal V}_{1}$;

\item[level $1$:] $X^{1} = Y^{1}$, $X^{1} \in X^{2}$, for $X^{1},Y^{1}
\in {\cal V}_{1}, X^{2} \in {\cal V}_{2}$;

\item[level $2$:] $T^{2} = U^{2}$, $T^{2} \in X^{3}$, where  $T^{2}$
and $U^{2}$ are terms of sort 2 and $X^{3} \in {\cal V}_{3}$.
\end{description}
$\QLQS$ \emph{purely universal formulae} are classified as:
\begin{description}
    \item[level $1$:] $(\forall z_{1}) \ldots (\forall z_{n})
    \varphi_{0}$, where $\varphi_{0}$ is any propositional combination
    of quantifier-free atomic formulae and $z_{1},\ldots,z_{n}$ are
    variables of sort $0$;

    \item[level $2$:] $(\forall Z_{1}^{1}) \ldots (\forall Z_{m}^{1})
    \varphi_{1}$, where $\varphi_{1}$ is any propositional combination
    of quantifier-free atomic formulae and of purely universal
    formulae of level 1, and $Z_{1}^{1},\ldots,Z_{m}^{1} \in
    \mathcal{V}_{1}$;

    \item[level $3$:] $(\forall Z_{1}^{2}) \ldots (\forall Z_{p}^{2})
    \varphi_{2}$, where $\varphi_{2}$ is any propositional combination
    of quantifier-free atomic formulae and of purely universal
    formulae of levels 1 and 2, and $Z_{1}^{2},\ldots,Z_{p}^{2} \in
    \mathcal{V}_{2}$.
\end{description}
Finally, the formulae of $\QLQS$ are all the propositional
combinations of quantifier-free atomic formulae of levels $0,1$, $2$,
and of purely universal formulae of levels $1$, $2$, $3$.

Next we introduce some notions that will be useful in the rest of the
paper.  Let $\varphi$ be a $\QLQS$-formula.  We can assume, without loss
of generality, that $\varphi$ contains as propositional connectives only
`$\neg$', `$\Or$', and `$\And$'.  Further, let $\syntree{\varphi}$ be the
syntax tree for $\varphi$ (see \cite{DJ90} for a precise definition), and
let $\nu$ be a node of $\syntree{\varphi}$.  We say that a
$\QLQS$-formula $\psi$ \emph{occurs} within $\varphi$ at position
$\nu$ if the subtree of $\syntree{\varphi}$ rooted at $\nu$ is identical
to $\syntree{\psi}$.  In this case we refer to $\nu$ as an
\emph{occurrence} of $\psi$ in $\varphi$ and to the path from the root
of $\syntree{\varphi}$ to $\nu$ as its \emph{occurrence path}.
An occurrence of a $\QLQS$-formula $\psi$ within a $\QLQS$-formula
$\varphi$ is \emph{positive} if its occurrence path deprived of its last
node contains an even number of nodes labelled by a $\QLQS$-formula of
type $\neg \chi$.  Otherwise, the occurrence is said to be
\emph{negative}.

\paragraph{Semantics of $\QLQS$.}
A {\em $\QLQS$-interpretation\/} is a pair $\model=(D,M)$, where $D$
is any \emph{nonempty} collection of objects, called the {\em
domain\/} or {\em universe\/} of $\model$, and $M$ is an assignment
over the variables of $\QLQS$ such that
    \begin{itemize}
        \item  $Mx \in D$, for each $x \in {\cal V}_{0}$;

        \item  $MX^{1} \in \pow(D)$, for each $X^{1} \in {\cal V}_{1}$;

        \item  $MX^{2} \in \pow(\pow(D))$, for each
        $X^{2} \in {\cal V}_{2}$;

        \item  $MX^{3} \in \pow(\pow(\pow(D)))$, for each
        $X^{3} \in {\cal V}_{3}$.\footnote{We recall that, for any set
    $s$, $\pow(s)$ denotes the \emph{powerset} of $s$, i.e., the
    collection of all subsets of $s$.}
    \end{itemize}
We assume that pair terms are interpreted \emph{\'{a} la} Kuratowski,
and therefore we put
\[
M\langle x,y\rangle \defAs \{\{Mx\},\{Mx,My\}\}\,.
\]
The introduction of a pairing operator in the language turned out to
be very useful in view of the applications in Section
\ref{sec:applications}.  Moreover, even if many pairing operations are
available (see for instance \cite{FoOmPo04}),
Kuratowski's style of encoding ordered pairs results to be quite
simple, at least for our purposes.

Let
\begin{itemize}
    \item[-]$\model=(D,M)$ be a $\QLQS$-interpretation,
    \item[-]  $x_{1},\ldots,x_{n} \in {\cal V}_{0}$,

    \item[-]  $X_{1}^{1},\ldots,X_{m}^{1} \in {\cal V}_{1}$,

    \item[-]  $X_{1}^{2},\ldots,X_{p}^{2} \in {\cal V}_{2}$,

    \item[-]  $u_{1},\ldots,u_{n} \in D$,

    \item[-]  $U_{1}^{1},\ldots,U_{m}^{1} \in \pow(D)$,

    \item[-]  $U_{1}^{2},\ldots,U_{p}^{2} \in \pow(\pow(D))$.
\end{itemize}
By $\model[x_{1}/u_{1},\ldots,x_{n}/u_{n},
         X_{1}^{1}/U_{1}^{1},\ldots,X_{m}^{1}/U_{m}^{1},
         X_{1}^{2}/U_{1}^{2},\ldots,X_{p}^2/U_{p}^{2}]\,$,
we denote the interpretation $\model' = (D,M')$ such that $M'x_{i}  =  u_{i}$, for $i = 1,\ldots,n$,
$M'X_{j}^{1} =  U_{j}^{1}$,  for $j = 1,\ldots,m$, $M'X_{k}^{2} =  U_{k}^{2}\,$, for $k = 1,\ldots,p$,
and which otherwise coincides with $\model$ on all remaining
variables.
Throughout the paper we use the abbreviations: ${\cal
M}^{z}$ for $\model[z_{1}/u_{1},\ldots,z_{n}/u_{n}]$, ${\cal
M}^{Z^1}$ for $\model[Z_{1}^{1}/U_{1}^{1},\ldots,Z_{m}^{1}/U_{m}^{1}]$, and ${\cal
M}^{Z^2}$ for $\model[Z_{1}^{2}/U_{1}^{2},\ldots,Z_{p}^2/U_{p}^{2}]$.

Let $\varphi$ be a $\QLQS$-formula and let $\model = (D, M)$ be a
$\QLQS$-interpre\-ta\-tion.  The notion of {\em satisfiability} of
$\varphi$ by $\model$ (denoted by $\model \models
\varphi$) is defined inductively over the structure of $\varphi$.
Quantifier-free atomic formulae are interpreted in the standard way
according to the usual meaning of the predicates `=' and `$\in$', and
purely universal formulae are evaluated as follows:
\begin{itemize}
\item[1.] $\model \models (\forall z_1) \ldots (\forall z_n) \varphi_0$
~~iff~~ $\model[z_1/u_1,\ldots , z_n/u_n] \models \varphi_0$, for
all $u_1,\ldots ,u_n \in D$;
\item[2.]  $\model \models (\forall Z_{1}^{1}) \ldots (\forall Z_{m}^{1})
\varphi_1$ iff $\model[Z_{1}^{1}/U_{1}^{1},\ldots , Z_{m}^{1}/U_{m}^{1}] \models
\varphi_1$, for all $U_{1}^{1},\ldots , U_{m}^{1} \in \pow(D)$;
\item[3.]  $\model \models (\forall Z_{1}^{2}) \ldots (\forall Z_{p}^{2})
\varphi_2$ iff $\model[Z_{1}^{2}/U_{1}^{2},\ldots , Z_{p}^{2}/U_{p}^{2}] \models
\varphi_2$, for all $U_{1}^{2},\ldots , U_{p}^{2} \in \pow(\pow(D))$.
\end{itemize}
Finally, evaluation of compound formulae follows the standard
rules of propositional logic.  If $\model \models \varphi$, i.e. $\model$
{\em satisfies\/} $\varphi$, then $\model$ is said to be a $\QLQS$-{\em
model\/} for $\varphi$.  A $\QLQS$-formula is said to be {\em
satisfiable\/} if it has a $\QLQS$-model.  A $\QLQS$-formula is
\emph{valid} if it is satisfied by all $\QLQS$-interpretations.

\subsection{Characterizing $\QLQSR$}\label{restrictionquant}
$\QLQSR$ is the subcollection of the formulae $\psi$ of $\QLQS$ for which the following restrictions hold.

\begin{enumerate}[{Restr.} I.]
\item \label{restriction1} For every purely universal formula $(\forall Z_1^1),\ldots,(\forall
Z_m^1)\varphi_1$ of level 2 occurring in $\psi$ and every purely universal formula $(\forall z_1) \ldots (\forall z_n)
\varphi_0$ of level
$1$ occurring negatively in $\varphi_1$, $\varphi_0$ is a propositional combination of level 0 quantifier-free atomic formulae and the
condition
\begin{equation}
    \label{condition}
\neg \varphi_0 \rightarrow \bigwedge_{i=1}^n \bigwedge_{j=1}^m  z_i \in Z_j^1
\end{equation}
is a valid $\QLQS$-formula (in this case we say that the formula
$(\forall z_1) \ldots (\forall z_n) \varphi_0$ is {\em linked} to the
variables $Z_1^1, \ldots , Z_m^1$).
\item \label{restriction2} For
every purely universal formula $(\forall Z_1^2),\ldots,(\forall Z_p^2)\varphi_2$ of level 3 occurring in $\psi$
\begin{itemize}
\item every purely universal formula of level 1 occurring negatively
in $\varphi_2$ and not occurring in a purely universal formula of
level 2, is only allowed to be of the form

    $(\forall z_1)\ldots(\forall z_n)\neg(\bigwedge_{i=1}^n
    \bigwedge_{j=1}^n\langle z_i,z_j\rangle=Y_{ij}^2)\,$, where
    $Y_{ij}^2 \in {\cal V}_2$, for $i,j = 1,\ldots,n$;

    \item purely universal formulae $(\forall Z_1^1),\ldots,(\forall
    Z_m^1)\varphi_1$ of level 2 may occur only positively in
    $\varphi_2$.
\end{itemize}

\end{enumerate}
Restriction~\ref{restriction1} is similar to the one described in \cite{CanNic08}.
In particular, following \cite{CanNic08}, we recall that condition~(\ref{condition})
guarantees that if a given interpretation assigns to $z_1,\ldots,z_n$
elements of the domain that make $\varphi_0$ false, then such elements
must be contained in the intersection of the sets assigned to
$Z_1^1,\ldots,Z_m^1$.  This fact is needed in the proof of statement
(ii) of Lemma \ref{quantifiedform} to make sure that satisfiability is
preserved in a suitable finite submodel (details, however, are not
reported here and can be found in \cite{CanNic08}).

Through several examples, in \cite{CanNic08} it is argued that
condition~(\ref{condition}) is not particularly restrictive.  Indeed,
to establish whether a given $\QLQS$-formula is a $\QLQSR$-formula,
since condition~(\ref{condition}) is a $2LS$-formula, its validity can
be checked using the decision procedure in \cite{FerOm1978}, as
$\QLQS$ is a conservative extension of $2LS$.  In addition, in many
cases of interest, condition~(\ref{condition}) is just an instance of
the simple propositional tautology $\neg( A \rightarrow B) \rightarrow
A$, and thus its validity can be established just by inspection.

Restriction~\ref{restriction2} has been introduced to be able to
express binary relations and several operations on them while keeping simple, at the same time, the decision procedure presented in Section \ref{decisionproc}.

Finally, we observe that though the semantics of $\QLQSR$ plainly
coincides with that of $\QLQS$, in what follows we prefer to
refer to $\QLQS$-interpretations of $\QLQSR$-formulae as
\emph{$\QLQSR$-interpretations}.

\section{The satisfiability problem for
$\QLQSR$-formulae}\label{satisfiability}

We will solve the satisfiability problem for $\QLQSR$, i.e. the
problem of establishing for any given formula of $\QLQSR$ whether it
is satisfiable or not, as follows:
\begin{enumerate}[(i)]
    \item \label{satisf1} firstly, we will show how to reduce
    effectively the satisfiability problem for $\QLQSR$-formulae to
    the satisfiability problem for \emph{normalized
    $\QLQSR$-conjunctions} (these will be defined shortly);

    \item \label{satisf2} secondly, we will prove that the collection
    of normalized $\QLQSR$-conjunctions enjoys a small model property.
\end{enumerate}
From (\ref{satisf1}) and (\ref{satisf2}), the solvability of the satisfiability problem for
$\QLQSR$ follows immediately.  Additionally, by further elaborating on
point (\ref{satisf1}), it could easily be shown that indeed the whole collection
of $\QLQSR$-formulae enjoys a small model property.

\subsection{Normalized $\QLQSR$-conjunctions}\label{normal3LQS}
Let $\psi$ be a formula of $\QLQSR$ and let $\psi_{DNF}$ be a
disjunctive normal form of $\psi$.  Then $\psi$ is satisfiable if
and only if at least one of the disjuncts of $\psi_{DNF}$ is
satisfiable. We recall that the disjuncts of $\psi_{DNF}$ are
conjunctions of literals, namely atomic formulae or their negation.\footnote{Atomic formulae are
quantified atomic formulae and purely universal formulae of any level.}
In view of the previous observations,
without loss of generality, we can suppose that our formula $\psi$
is a conjunction of level $0,1,2$ quantifier-free literals and
of level $1,2,3$ quantified literals.  In addition, we
can also assume that no variable occurs both bound and free in $\psi$
and that distinct occurrences of quantifiers bind distinct
variables.

For decidability purposes, negative quantified conjuncts occurring in
$\psi$ can be eliminated as follows.  Let $\model = (D, M)$ be a model
for $\psi$, and let $\neg (\forall z_1) \ldots (\forall z_n)
\varphi_0$ be a negative quantified literal of level 1 occurring in
$\psi$.  Since $\model \models \neg (\forall z_1) \ldots (\forall z_n)
\varphi_0$ if and only if $\model[z_1/u_1,\ldots , z_n/u_n] \models
\neg \varphi_0$, for some $u_1,\ldots ,u_n \in D$, we can replace
$\neg (\forall z_1) \ldots (\forall z_n) \varphi_0$ in $\psi$ by $\neg
(\varphi_{0})^{z_{1},\ldots,z_{n}}_{z'_{1},\ldots,z'_{n}}$, where
$z'_{1},\ldots,z'_{n}$ are newly introduced variables of sort 0.
Negative quantified literals of levels 2 and 3 can be dealt with much
in the same way and hence, we can further assume that $\psi$ is a
conjunction of literals of the following types:

\begin{itemize}
\item[(1)] quantifier-free literals of any level;

\item[(2)] purely universal formulae of level 1;

\item[(3)] purely universal formulae of level 2 and 3 satisfying
Restrictions \ref{restriction1} and \ref{restriction2} given in
Section \ref{restrictionquant}, respectively.
\end{itemize}
We call these formulae {\em normalized $\QLQSR$-conjunctions}.

\subsection{A small model property for normalized $\QLQSR$-conjunctions}
\label{decisionproc}
In view of the above reductions, we can limit ourselves to consider
the satisfiability problem for normalized $\QLQSR$-conjunctions only.
Thus, let $\psi$ be a normalized $\QLQSR$-conjunction and assume that
$\model = (D, M)$ is a model for $\psi$.

We show how to construct, out of the model $\model$, a finite
$\QLQSR$-interpretation $\model^* = (D^*,M^*)$ which is a model of
$\psi$ sufficiently rich to reconstruct any possible counter-example to the formula and
such that the size of $D^*$ depends solely on the size of
$\psi$.  We will proceed as follows.  First, in Section \ref{ssseUniv}, we outline a procedure for
the construction of a nonempty finite universe $D^* \subseteq D$. In Steps 1 to 3 $D^*$
is provided with enough elements to properly interpret quantifier-free atomic formulae. Cases involving
variables of levels 2 and 3 are treated in Step 2 by introducing an additional set of new variables, ${\cal V}_{1}^{F}$.
Finally, in Step 4 $D^*$ is further enriched to take care of purely universal formulae of level 2.
Then we show how to relativize $\model$ to $D^*$ according to
Definition \ref{relintrp} below, thus defining a finite
$\QLQSR$-interpretation $\model^* = (D^*,M^*)$.  Finally, we prove
that $\model^*$ satisfies $\psi$.

\subsubsection{Construction of the universe $D^*$}
\label{ssseUniv}
Let us denote by ${\cal V}_0'$, ${\cal V}_1'$, and ${\cal V}_2'$ the
collections of variables of sort $0$, $1$, and $2$ occurring free in
$\psi$, respectively.  We construct $D^*$ according to the following
steps:
\begin{description}
\item [Step 1:] Let $\mathcal{F} = \mathcal{F}_1 \cup \mathcal{F}_2$,
where
\begin{itemize}
\item $\mathcal{F}_1$ `distinguishes' the set $S = \{MX^{2} : X^{2}
\in {\cal V}_2'\}$, in the sense that $K \cap \mathcal{F}_1 \neq K'
\cap \mathcal{F}_1$ for every distinct $K,K' \in S$.  Such a set
$\mathcal{F}_1$ can be constructed by the procedure {\em Distinguish}
described in \cite{CanFer1995}.  As shown in \cite{CanFer1995}, we can
also assume that $|\mathcal{F}_1| \leq |S| -1$.

\item $\mathcal{F}_2$ satisfies $|MX^{2} \cap \mathcal{F}_2| \geq
\min(3, |MX^{2}|)$, for every $X^{2} \in {\cal V}_2'$. Plainly, we
can also assume that $|\mathcal{F}_2| \leq 3 \cdot |{\cal V}_2'|$.
\end{itemize}

\item [Step 2:] Let $\{F_1,\ldots,F_k\} = \mathcal{F} \setminus \{MX^1
: X^1 \in {\cal V}_1'\}$ and let ${\cal V}_{1}^{F} =
\{X_{1}^{1},\ldots,X_{k}^{1}\}\subseteq {\cal V}_1$ be such that
${\cal V}_{1}^{F} \cap {\cal V}_1' = \emptyset$ and ${\cal V}_{1}^{F}
\cap {\cal V}_{1}^{B} = \emptyset$, where ${\cal V}_{1}^{B}$ is the
collection of bound variables in $\psi$.  Let $\overline{\model}$ be
the interpretation $\model[X_{1}^{1}/F_1,\ldots, X_{k}^{1}/F_k]$.
Since the variables in ${\cal V}_{1}^{F}$ do not occur in $\psi$
(neither free nor bound), their evaluation is immaterial for $\psi$
and therefore, from now on, we identify $\overline{\model}$ and
$\model$.

\item [Step 3:] Let $\Delta = \Delta_{1} \cup \Delta_{2}$, where
\begin{itemize}
\item $\Delta_{1}$ distinguishes the set $T = \{MX^1 : X^1 \in {\cal
V}_1'\cup {\cal V}_1^F\}$ and $|\Delta_{1}| \leq |T| -1$ holds (cf.\
Step 1 above);

\item $\Delta_{2}$ satisfies $|J \cap \Delta_{2}| \geq \min(3, |J|)$,
for every $J \in \{MX^1 : X^1 \in {\cal V}_1'\cup {\cal V}_1^F\}$.
Plainly, we can assume that $|\Delta_{2}| \leq 3 \cdot |{\cal
V}_1'\cup {\cal V}_1^F|$.

We initialize $D^*$ by putting
$$
D^* :=\{Mx : x \hbox{ in } {\cal V}_0'\} \cup \Delta\,.
$$
\emph{($D^*$ will possibly be enlarged during the subsequent Step 4.)}
\end{itemize}

\item [Step 4:] Let $\chi_1, \ldots , \chi_r$ be all the purely universal
formulae of level 2 occurring in $\psi$. To each conjunct
$\chi_i \equiv (\forall Z_{i,h_1}^{1})
\ldots (\forall Z_{i,h_{m_i}}^{1}) \varphi_i$, we associate the
collection $\varphi_{i,k_1}, \ldots , \varphi_{i,k_{\ell_i}}$ of atomic
formulae of the form $(\forall z_1) \ldots (\forall z_n) \varphi_0$
present in the matrix of $\chi_i$, and call the variables
$Z_{i,h_1}^{1}, \ldots , Z_{i,h_{m_i}}^{1}$ the {\em arguments of }
$\varphi_{i,k_{1}}, \ldots , \varphi_{i,k_{\ell_i}}$.

Let us put
$$
\Phi \defAs \{ \varphi_{i,k_j} : 1 \leq j  \leq \ell_i \hbox{ and }
 1 \leq i \leq r  \}.
$$
Then, for each $\varphi \in \Phi$ of the form $(\forall z_1) \ldots
(\forall z_n) \varphi_0$ having $Z_1^{1},\ldots,Z_m^{1}$ as arguments,
and for each ordered $m$-tuple $(X_{h_1}^{1},\ldots,X_{h_m}^{1})$ of
variables in ${\cal V}_1'\cup {\cal V}_1^F$, if
$M({\varphi_{0}})_{X_{h_1}^{1},\ldots,X_{h_m}^{1}}^{Z_1^{1}\;\,,\ldots,\;Z_m^{1}}
= \hbox{\bf false}$ we insert in $D^*$ elements $u_1,\ldots,u_n \in D$
such that
$$M[z_1/u_1,\ldots,z_n/u_n]
(\varphi_{0})_{X_{h_1}^{1},\ldots,X_{h_m}^{1}}^{Z_1^{1}\;\,,\ldots,\;Z_m^{1}}
=\false\,,$$ otherwise we leave $D^*$ unchanged.
\end{description}

Next, we calculate a bound to the size of $D^*$.  Since
$|\mathcal{F}_1| \leq |S| - 1 \leq |{\cal V}_2'|-1$ and
$|\mathcal{F}_2|\leq 3|{\cal V}_2'|$ (cf.\ Step 1 above), we plainly
have $|\mathcal{F}|\leq 4|{\cal V}_2'| - 1$.  Analogously, just after
Step 3, we have $|\Delta|\leq 4(|{\cal V}_1'| + (4|{\cal V}_2'| -1))
-1$ and $|D^*| \leq |{\cal V}_0'| + 4|{\cal V}_1'| + 16|{\cal V}_2'|
-5$.  Finally, after Step 4, if we let $L_{m}$ denote the maximal
length of the quantifier prefix of any purely universal formula of
level 2 occurring in $\psi$, and $L_{n}$ denote the maximal length of
the quantifier prefix of $\varphi \equiv (\forall z_1)\ldots(\forall
z_n)\varphi_0$, with $\varphi$ ranging in $\Phi$, then we have
\begin{equation}
    \label{DStar}
    |D^*| \leq |{\cal V}_0'| + 4|{\cal V}_1'| + 16|{\cal V}_2'| +
    \left((|{\cal V}_1'| + 4|{\cal V}_2'| - 1)^{L_{m}} L_{n}\right)|\Phi| - 5\,.
\end{equation}
Thus, it turns out that, in general, the domain $D^*$ (of the small
model) is exponential in the size of the input formula $\psi$.

\subsubsection{Relativized interpretations}
\label{machinery}
We introduce now the notion of \emph{relativized interpretation},
whose domain is the set $D^{*}$ constructed above, to define, out of
the model $\model = (D,M)$ of our normalized $\QLQSR$-conjunction
$\psi$, a finite interpretation $\model^* = (D^*,M^*)$ of bounded
size, which also satisfies $\psi$.

\begin{definition}\label{relintrp}
Let $\model=(D,M)$, $D^{*}$, ${\cal V}'_{1},{\cal V}_1^F$, and ${\cal
V}'_{2}$ be as above, and let $d^{*} \in D^{*}$.  The {\em
relativized\/} interpretation $\model^* =\Rel(\model, D^{*}, d^{*}, {\cal
V}'_{1}, {\cal V}_1^F, {\cal V}'_{2})$ of $\model$ with respect to
$D^{*}$, $d^{*}$, ${\cal V}'_{1}$, ${\cal V}_1^F$, and ${\cal V}'_{2}$
is the $\QLQSR$-interpretation $(D^{*},M^{*})$ such that
\begin{eqnarray*}
    M^{*}x & = & \left\{
    \begin{array}{ll}
        Mx\,, & \mbox{if $Mx \in D^{*}$}  \\
        d^{*}\,, & \mbox{otherwise}\, ,
    \end{array}
    \right.
    \\
    M^{*}X^1 & = & MX^1 \cap D^{*}\, ,
    \\
   M^{*}X^2  &=&
    \left((MX^2 \cap \pow(D^{*})) \setminus \{M^{*}X^1: X^1 \in
    ({\cal V}'_{1}\cup {\cal V}_1^F)\}\right)
    \\
     & & \qquad\cup \{M^{*}X^1: X^1 \in ({\cal V}'_{1} \cup {\cal V}_1^F),~MX^1 \in MX^2\}\, ,\\
M^{*}X^3  &=&
    \left((MX^3 \cap \pow(\pow(D^{*}))) \setminus \{M^{*}X^2: X^2 \in
    {\cal V}'_{2}\}\right)
     \\
     & & \qquad\cup \{M^{*}X^2: X^2 \in {\cal V}'_{2},~MX^2 \in MX^3\}\,.
\end{eqnarray*}
\end{definition}
Concerning $M^*X^2$ and $M^*X^3$, we observe that they have been
defined in such a way that all the membership relations between
variables of $\psi$ of sorts 2 and 3 are the same in both the
interpretations $\model$ and $\model^*$.  This fact will be proved
in the next section.

For ease of notation, we will often omit the reference to the element
$d^{*} \in D^{*}$ and write simply\/ $\Rel(\model, D^{*}, {\cal
V}'_{1}, {\cal V}_1^F, {\cal V}'_{2})$ in place of\/ $\Rel(\model,
D^{*}, d^{*}, {\cal V}'_{1}, {\cal V}_1^F, {\cal
V}'_{2})$, when $d^{*}$ is clear from the context.

The following useful properties are immediate consequences of the
construction of $D^*$, for any $x,y \in {\cal V}'_0$, $X^1,Y^1 \in
{\cal V}'_1$, and $X^2, Y^2 \in {\cal V}'_2$:
\begin{enumerate}[(A)]
    \item \label{prop1} if $MX^1 \neq MY^1$, then $(MX^1
    \bigtriangleup MY^1) \cap D^{*} \neq \emptyset$,\footnote{We
    recall that for any sets $s$ and $t$, $s \bigtriangleup t$ denotes
    the symmetric difference of $s$ and of $t$, namely the
    set $(s \setminus t)\cup(t \setminus s)$.}

    \item \label{prop2} if $MX^2 \neq MY^2$, there is a $J\in (MX^2
    \bigtriangleup MY^2)\cap \{MX^1 : X^1 \in ({\cal V}'_{1}\cup {\cal
    V}'_{F})\}$ such that $J \cap D^* \neq \emptyset$,

    \item \label{prop3} if $M\langle x,y \rangle \neq MX^2$, there is
    a $J \in (MX^2 \bigtriangleup M\langle x,y \rangle) \cap \{MX^1 :
    X^1 \in ({\cal V}'_{1}\cup {\cal V}'_{F})\}$ such that $J \cap D^*
    \neq \emptyset$, and if $J \in MX^2$, then $J \cap D^* \neq
    \{Mx\}$ and $J \cap D^* \neq \{Mx,My\}$.
\end{enumerate}

\subsection{Soundness of the relativization}
As above, let $\model=(D,M)$ be a $\QLQSR$-interpretation satisfying
our given normalized $\QLQSR$-conjunction $\psi$, and let $D^{*}$,
${\cal V}'_{1}$, ${\cal V}_1^F$, ${\cal V}'_{2}$, and $\model^{*}$ be
defined as before.  The main result of this section is Theorem
\ref{correctness} which states that if $\model$ satisfies $\psi$, then
$\model^*$ satisfies $\psi$ as well.  The proof of Theorem
\ref{correctness} exploits the technical Lemmas \ref{le_basic},
\ref{le_M*zMz*}, \ref{le_eqM*ZMZ*1}, \ref{le_eqM*ZMZ*2}, and
\ref{quantifiedform} below.  In particular, Lemma \ref{le_basic}
states that $\model$ satisfies a quantifier-free atomic formula
$\varphi$, fulfilling conditions (\ref{prop1}), (\ref{prop2}), and
(\ref{prop3}) above, if and only if $\model^*$ satisfies $\varphi$ too.
Lemmas \ref{le_M*zMz*}, \ref{le_eqM*ZMZ*1}, and \ref{le_eqM*ZMZ*2}
claim that suitably constructed variants of $\model^*$ and the small
models resulting by applying the construction of Section
\ref{decisionproc} to the corresponding variants of $\model$ can be
considered identical.
Finally, Lemma \ref{quantifiedform}, which follows from Lemmas
\ref{le_basic}, \ref{le_M*zMz*}, \ref{le_eqM*ZMZ*1}, and
\ref{le_eqM*ZMZ*2}, states that $\model^*$ satisfies all quantified
conjuncts of $\psi$ which are satisfied by $\model$.

\begin{lemma}
\label{le_basic} The following statements hold:
\begin{enumerate}[(a)]
    \item \label{le1} $\model^{*} \models x = y$ iff $\model \models x
    = y$, for all $x,y \in {\cal V}_{0}$ such that $Mx, My \in D^{*}$;

    \item \label{le2} $\model^{*} \models x \in X^1$ iff $\model
    \models x \in X^1$, for all $X^1 \in {\cal V}_{1}$ and $x \in
    {\cal V}_{0}$ such that $Mx \in D^{*}$;

    \item \label{le3} $\model^{*} \models X^1 = Y^1$ iff $\model
    \models X^1 = Y^1$, for all $X^1,Y^1 \in {\cal V}_{1}$ such that
    condition (\ref{prop1}) holds;

    \item \label{le4} $\model^{*} \models X^1 \in X^2$ iff $\model
    \models X^1 \in X^2$, for all $X^1 \in ({\cal V}'_{1}\cup {\cal
    V}'_{F})$, $X^2 \in {\cal V}_{2}$;

    \item \label{le5} $\model^{*} \models X^2 = Y^2$ iff $\model
    \models X^2 = Y^2$, for all $X^2,Y^2 \in {\cal V}_{2}$ such that
    condition (\ref{prop2}) holds;

    \item \label{le6} $\model^{*} \models \langle x,y \rangle = X^2$
    iff $\model \models \langle x,y \rangle = X^2$, for all $x,y \in
    {\cal V}_{0}$ such that $Mx,My \in D^*$ and $X^2 \in {\cal V}_{2}$
    such that condition (\ref{prop3}) holds;

    \item \label{le7} $\model^{*} \models \langle x,y \rangle \in X^3$
    iff $\model \models \langle x,y \rangle \in X^3$, for all $x,y \in
    {\cal V}_{0}$ such that $Mx,My \in D^*$ and $X^2 \in {\cal V}_{2}$
    such that condition (\ref{prop3}) holds;

    \item \label{le8} $\model^{*} \models X^2 \in X^3$ iff $\model
    \models X^2 \in X^3$, for all $x,y \in {\cal V}_{0}$ such that
    $Mx,My \in D^*$ and $X^2 \in {\cal V}_{2}$ such that conditions
    (\ref{prop2}) and (\ref{prop3}) hold.
\end{enumerate}
\end{lemma}
\begin{proof}
\begin{itemize}
\item [(\ref{le1})] Let $x,y \in {\cal V}_{0}$ be such that $Mx, My
\in D^{*}$.  Then $M^{*}x = Mx$ and $M^{*}y = My$, so we have
immediately that ${\cal M}^{*} \models x = y$ ~iff~ ${\cal M} \models x
= y$.

\item [(\ref{le2})] Let $X^1 \in {\cal V}_{1}$ and let $x \in {\cal
V}_{0}$ be such that $Mx \in D^{*}$.  Then $M^{*}x = Mx$, so that
$M^{*}x \in M^{*}X^1$ ~iff~ $Mx \in MX^1 \cap D^{*}$ iff $Mx \in MX^1$.

\item [(\ref{le3})] If $MX^1 = MY^1$, then plainly $M^{*}X^1 =
M^{*}Y^1$.  On the other hand, if $MX^1 \neq MY^1$, then, by condition
(\ref{prop1}), $(MX^1 \bigtriangleup MY^1) \cap D^{*} \neq \emptyset$ and thus
$M^{*}X^1 \neq M^{*}Y^1$.

\item [(\ref{le4})] If $MX^1 \in MX^2$, then $M^*X^1 \in M^*X^2$.  On
the other hand, suppose by contradiction that $MX^1 \notin MX^2$ and
$M^*X^1 \in M^*X^2$.  Then, there must necessarily be a $Z^1 \in
({\cal V}_{1}' \cup {\cal V}_{1}^{F})$ such that $MZ^1 \in MX^2$,
$MZ^1 \neq MX^1$, and $M^* X^1 = M^* Z^1$.  Since $MZ^1 \neq MX^1$ and
$(MZ^1 \bigtriangleup MX^1) \cap D^{*} \neq \emptyset$, by condition
(\ref{prop1}), we have $M^* X^1 \neq M^* Z^1$, which is a
contradiction.

\item [(\ref{le5})] If $MX^2 = MY^2$, then $M^* X^2 = M^* Y^2$.  On
the other hand, if $MX^2 \neq MY^2$, by condition (\ref{prop2}), there
is a $J \in (MX^2 \bigtriangleup MY^2) \cap \{MX^1 : X^1 \in ({\cal
V}_{1}'\cup {\cal V}_{1}^{F})\}$ such that $J \cap D^* \neq
\emptyset$.  Let $J = MX^1$, for some $X^1 \in ({\cal V}_{1}'\cup
{\cal V}_{1}^{F})$, and suppose without loss of generality that $MX^1
\in MX^2$ and $MX^1 \notin MY^2$.  Then, by (\ref{le4}), $M^*X^1 \in
M^* X^2$ and $M^* X^1 \notin M^* Y^2$, and hence $M^* X^2 \neq M^*
Y^2$.

\item [(\ref{le6})] If $M\langle x,y \rangle = MX^2$, then $M^*
\langle x,y \rangle = M^* X^2$.  If $M\langle x,y \rangle \neq MX^2$,
then there is a $J \in (MX^2 \bigtriangleup M\langle x,y \rangle) \cap
\{MX^1 : X^1 \in ({\cal V}_{1}' \cup {\cal V}_{1}^{F})\}$ satisfying
the constraints of condition (\ref{prop3}).  Let $J = MX^1$, for some
$X^1 \in ({\cal V}_{1}'\cup {\cal V}_{1}^{F})$, and suppose that $MX^1
\in MX^2$ and $MX^1 \notin M\langle x,y\rangle$.  Then $M^*X^1 \in
M^*X^2$ and since $M^*X^1 \neq \{Mx\}$ and $M^*X^1 \neq \{Mx,My\}$, it
follows that $M^*X^1 \notin M^*\langle x,y \rangle$.  On the other
hand, if $MX^1 \in M\langle x,y\rangle$ and $MX^1 \notin MX^2$, then
either $MX^1 = \{Mx\}$ or $MX^1 = \{Mx,My\}$.  In both cases $MX^1 =
M^*X^1$ and thus if $MX^1 \notin MX^2$, it plainly follows that
$M^*X^1 \notin M^*X^2$.

\item [(\ref{le7})] Let $x,y \in {\cal V}_{0}$ and $X^3 \in {\cal
V}_{3}$ be such that $M\langle x,y\rangle \in MX^3$.  Then $M^*
\langle x,y\rangle \in M^*X^3$.  On the other hand, suppose by
contradiction that $M\langle x,y\rangle \notin MX^3$ and $M^* \langle
x,y\rangle \in M^*X^3$.  Then, there must be an $X^2 \in {\cal
V}_{2}'$ such that $M^*X^2 \in M^* X^3$, $M^* X^2 = M^*\langle
x,y\rangle$, and $MX^2 \neq M\langle x,y\rangle$.  But this is
impossible by (\ref{le6}).

\item [(\ref{le8})] If $MX^2 \in MX^3$ then $M^*X^2 \in M^*X^3$.  Now
suppose by contradiction that $MX^2 \notin MX^3$ and that $M^*X^2 \in
M^*X^3$.  Then, either there is a $Y^2 \in {\cal V}_2'$ such that
$MX^2 \neq MY^2$ and $M^*X^2 = M^*Y^2$, which is impossible by
(\ref{le5}), or there is a $\langle x,y\rangle$, with $x,y \in {\cal
V}_{0}$, $Mx,My \in D^*$, such that $MX^2 \neq M\langle x,y\rangle$
and $M^*X^2 = M^*\langle x,y\rangle$, but this is absurd by
(\ref{le6}).
\end{itemize}
\end{proof}

In view of the next technical lemmas, we introduce the following
notations.  Let $u_1,\ldots,u_n\in D^*$,
$U_{1}^{1},\ldots,U_{m}^{1}\in \pow(D^{*})$, and
$U_{1}^{2},\ldots,U_{p}^{2}\in \pow(\pow(D^{*}))$.  Then we put
\begin{eqnarray*}
    \model^{*,z} & = & \model^{*}[z_{1}/u_{1},\ldots,z_{n}/u_{n}], \\
    \model^{*,Z^1} & = & \model^{*}[Z_{1}^{1}/U_{1}^{1},\ldots,Z_{m}^{1}/U_{m}^{1}], \\
    \model^{*,Z^2} & = & \model^{*}[Z_{1}^{2}/U_{1}^{2},\ldots,Z_{p}^{2}/U_{p}^{2}],
\end{eqnarray*}
and also
\begin{eqnarray*}
    \model^{z,*} & = & \Rel(\model^{z},D^{*}, {\cal V}'_{1}, {\cal V}_{1}^{F}, {\cal V}'_{2}),\\
    \model^{Z^1,*} & = & \Rel(\model^{Z^{1}},D^{*},{\cal V}'_{1} \cup
                         \{Z_{1}^{1},\ldots,Z_{m}^{1}\}, {\cal V}_{1}^{F}, {\cal V}'_{2}),\\
    \model^{Z^2,*} & = & \Rel(\model^{Z^{2}},D^{*},\mathcal{F}^{*},{\cal V}'_{1}, {\cal V}_{1}^{F}, {\cal V}'_{2}\cup
                         \{Z_{1}^{2},\ldots,Z_{p}^{2}\}).
\end{eqnarray*}
The next three lemmas claim that, under certain conditions, the
following pairs of $\QLQSR$-interpretations $\model^{*,z}$ and
$\model^{z,*}$, $\model^{*,Z^1}$ and $\model^{Z^1,*}$,
$\model^{*,Z^2}$ and $\model^{Z^2,*}$ can be
identified, respectively.

\begin{lemma}
\label{le_M*zMz*} Let $u_{1},\ldots,u_{n} \in D^{*}$, and let
$z_{1},\ldots,z_{n} \in {\cal V}_{0}$. Then, the $\QLQSR$-interpretations $\model^{*,z}$
and $\model^{z,*}$ coincide.

\end{lemma}
\begin{proof}
The proof of the lemma is carried out by showing that $\model^{*,z}$ and $\model^{z,*}$ agree over variables of all sorts.
\begin{itemize}
\item Let $x \in {\cal V}_{0}$. Since $u_{1},\ldots,u_{n} \in D^{*}$, the thesis  follows immediately.

\item Let $X^1 \in {\cal V}_{1}$, then $M^{*,z}X^1 = M^{*}X^1 =
MX^1 \cap D^{*} = M^{z}X^1 \cap D^{*} = M^{z,*}X^1$.

\item Let $X^2 \in {\cal V}_{2}$, then we have the following equalities:
\begin{eqnarray*}
        M^{*,z}X^2 & = & M^{*}X^2 =
                        ((MX^2 \cap \pow(D^{*})) \setminus \{M^{*}X^1: X^1 \in ({\cal V}'_{1}\cup {\cal V}_{1}^{F})\})\\
                 &   &  \phantom{M^{*}X^2 =} \qquad \cup \{M^{*}X^1: X^1 \in ({\cal V}'_{1}\cup {\cal V}_{1}^{F}),~MX^1 \in MX^2\} \,  \\
                 &  & \phantom{M^{*}X^2} = ((M^{z}X^2 \cap \pow(D^*)) \setminus \{M^{z,*}X^1 : X^1 \in ({\cal V}_{1}'\cup {\cal V}_{1}^{F})\}) \\
                 &   &  \phantom{M^{*}X^2 =} \qquad \cup \{M^{z,*}X^1 : X^1 \in ({\cal V}_{1}'\cup {\cal V}_{1}^{F}), M^{z}X^1 \in M^{z}X^2\} \\
                 &  & \phantom{M^{*}X^2} =  M^{z,*}X^2 \,.
\end{eqnarray*}
\item Let $X^3 \in {\cal V}_{3}$, then the following holds:
\begin{eqnarray*}
        M^{*,z}X^3 & = & M^{*}X^3 =
                        ((MX^3 \cap \pow(\pow(D^*))) \setminus \{M^{*}X^2 : X^2 \in {\cal V}_{2}'\})\\
                 &   &  \phantom{M^{*}X^2 =} \qquad\cup \{M^{*}X^2 : X^2 \in {\cal V}_{2}', MX^2 \in MX^3\}  \\
                 &  & \phantom{M^{*}X^2} = ((M^{z}X^3 \cap \pow(\pow(D^*))) \setminus \{M^{z,*}X^2 : X^2 \in {\cal V}_{2}'\}) \\
                 &   &  \phantom{M^{*}X^2 =} \qquad\cup \{M^{z,*}X^2 : X^2 \in {\cal V}_{2}', M^{z}X^2 \in M^{z}X^3\} \\
                 &  & \phantom{M^{*}X^2} =  M^{z,*}X^3 \,.
\end{eqnarray*}
\end{itemize}
\end{proof}

\begin{lemma}
\label{le_eqM*ZMZ*1}
Let $Z_{1}^1,\ldots,Z_{m}^1 \in {\cal V}_{1} \setminus ({\cal V}'_{1}
\cup {\cal V}_{1}^{F})$ and $U_{1}^1,\ldots,U_{m}^1 \in \pow(D^*)
\setminus \{M^{*}X^1 : X^1 \in ({\cal V}'_{1} \cup {\cal
V}_{1}^{F})\}$.  Then, the $\QLQSR$-interpretations $\model^{*,Z^1}$
and $\model^{Z^1,*}$ coincide.
\end{lemma}
\begin{proof}
We prove the lemma by showing that $\model^{*,Z^1}$ and $\model^{Z^1,*}$ agree over variables of all sorts.

\begin{enumerate}
\item  Clearly $M^{*,Z^1}x = M^{*}x = M^{Z^1,*}x$, for
    all individual variables $x \in {\cal V}_{0}$.

\item
    Let $X^1 \in {\cal V}_{1}$.
    If $X^1 \notin \{Z_{1}^1,\ldots,Z_{m}^1\}$, then
    \[
    M^{Z^1,*}X^1 = M^{Z^1}X^1 \cap D^{*} = MX^1 \cap D^{*} = M^{*}X^1 = M^{*,Z^1}X^1\,.
    \]
    On the other hand, if $X^1 = Z_j^1$ for some $j \in \{1,\ldots, m\}$, we have
    \[
    M^{Z^1,*}Z_{j}^1 = M^{Z^1}Z_{j}^1 \cap D^{*} = U_{j}^1 \cap D^{*} = U_{j}^1
      =  M^{*,Z^1}Z_{j}^1\,.
    \]
\item
    Let $X^2 \in {\cal V}_{2}$.
    Then we have
    \begin{eqnarray}
        M^{*,Z^1}X^2 & = & M^{*}X^2 =
                        ((MX^2 \cap \pow(D^{*})) \setminus \{M^{*}X^1:
			X^1 \in ({\cal V}'_{1}\cup {\cal
			V}_{1}^{F})\}) \notag\\
                 &   &  \phantom{M^{*}A =} \qquad\cup \{M^{*}X^1: X^1
         \in ({\cal V}'_{1}\cup {\cal V}_{1}^{F}),~MX^1 \in
         MX^2\} \, , \label{a1}\\[.2cm]
    M^{Z^1,*}X^2 & = & ((M^{Z^1}X^2 \cap \pow(D^{*})) \setminus
                        \{M^{Z^1,*}X^1: X^1 \in (({\cal V}'_{1}\cup {\cal V}_{1}^{F}) \cup
                             \{Z_{1}^1,\ldots,Z_{m}^1\})\}) \notag\\
                 &   &  \qquad\cup \{M^{Z^1,*}X^1: X^1 \in (({\cal V}'_{1}\cup {\cal V}_{1}^{F}) \cup
                             \{Z_{1}^1,\ldots,Z_{m}^1\}),~M^{Z^1}X^1 \in M^{Z^1}X^2\} \notag\\
                 & = & ((MX^2 \cap \pow(D^{*})) \setminus
                        (\{M^{*}X^1: X^1 \in ({\cal V}'_{1}\cup {\cal V}_{1}^{F})\}  \cup
                        \{U_{j}: j = 1,\ldots,m \})) \notag\\
                 &   &  \qquad\cup (\{M^{*}X^1: X^1 \in ({\cal V}'_{1}\cup {\cal V}_{1}^{F}),~
		 MX^1 \in MX^2\}\notag\\
                 &   &  \qquad\cup (\{U_{j}: j = 1,\ldots,m\} \cap
		 MX^2))\,.\label{a2}
    \end{eqnarray}

    By putting
    \[
    \begin{array}{lll}
        P_{1} & = & MX^2 \cap \pow(D^{*}) ,  \\
        P_{2} & = & \{M^{*}X^1: X^1 \in ({\cal V}'_{1}\cup {\cal V}_{1}^{F})\} , \\
        P_{3} & = & \{U_{j}: j = 1,\ldots,m \},  \\
        P_{4} & = & \{M^{*}X^1: X^1 \in ({\cal V}'_{1}\cup {\cal V}_{1}^{F}),~MX^1 \in MX^2\},  \\
        P_{5} & = & \{U_{j}: j = 1,\ldots,m\} \cap MX^2,
    \end{array}
    \]
    then by (\ref{a1}) and (\ref{a2}) can be rewritten as
    \begin{eqnarray}
        M^{*,Z^1}X^2 & = & (P_{1} \setminus P_{2}) \cup P_4 \label{a3}\\
        M^{Z^1,*}X^2 & = & (P_{1} \setminus (P_{2} \cup P_3))
                       \cup P_4 \cup P_5\,.\label{a4}
    \end{eqnarray}
    Moreover, since, as can easily verified, we have
    \[
        P_{2} \cap P_{3}  =  \emptyset\,,~~~
        P_{5}  =  P_{1} \cap P_{3}\,,~~~ \text{and~~~}
        P_{4} \subseteq  P_{2}\,,
    \]
    then
    \begin{eqnarray*}
   (P_{1} \setminus P_{2}) \cup P_4 &=&
    (P_{1} \setminus (P_{2} \cup P_3))
     \cup P_4 \cup (P_{1} \cap P_{3}) \\
    &=&  (P_{1} \setminus (P_{2} \cup P_3))
     \cup P_4 \cup P_5\,.
     \end{eqnarray*}
    Therefore, (\ref{a3}) and (\ref{a4}) readily imply $M^{*,Z^1}X^2 = M^{Z^1,*}X^2$.
\item
    Let $X^3 \in {\cal V}_{3}$, then $M^{*,Z^1} X^3 = M^*[Z_1^1/U_1^1,\ldots,Z_m^1/U_m^1]X^3 = M^* X^3$ and
    \begin{eqnarray*}
        M^{Z^1,*} X^3 & = & ((M^{Z^1} X^3 \cap \pow(\pow(D^*))) \setminus \{M^{Z^1,*}X^2 : X^2 \in {\cal V}'_{2}\})\\
                 &   &  \qquad\cup \{M^{Z^1,*}X^2 : X^2 \in {\cal V}'_{2}, M^{Z^1}X^2 \in M^{Z^1}X^3\} \\
                 & = & ((MX^3 \cap \pow(\pow(D^*))) \setminus \{M^{*}X^2 : X^2 \in {\cal V}'_{2}\}) \\
                 &   &  \qquad\cup \{M^{*}X^2 : X^2 \in {\cal V}'_{2}, MX^2 \in MX^3\}\\
                 & = & M^{*}X^3 \,.
    \end{eqnarray*}
    Since $M^{*,Z^1}X^3 = M^{Z^1,*}X^3$, the thesis follows.
\end{enumerate}
\end{proof}

\begin{lemma}
\label{le_eqM*ZMZ*2}
Let $Z_{1}^2,\ldots,Z_{p}^2 \in {\cal V}_{2} \setminus {\cal V}'_{2}$
and $U_{1}^2,\ldots,U_{p}^2 \in \pow(\pow(D^*)) \setminus \{M^{*}X^2 :
X^2 \in {\cal V}'_{2}\}$.  Then the $\QLQSR$-interpretations
$\model^{*,Z^2}$ and $\model^{Z^2,*}$ coincide.
\end{lemma}
\begin{proof}
We show that $\model^{*,Z^2}$ and $\model^{Z^2,*}$
coincide by proving that they agree over variables of all sorts.
\begin{enumerate}
\item Plainly $M^{*,Z^2}x = M^{*}x = M^{Z^2,*}x$, for every $x \in
{\cal V}_0$.

\item Let $X^1 \in {\cal V}_1$. Then $M^{*,Z^2}X^1 = M^{*}X^1 =
M^{Z^2,*}X^1$.

\item  Let $X^2 \in {\cal V}_2$ such that $X^2 \notin
\{Z_1^2,\ldots,Z_p^2\}$. Then
\[M^{*,Z^2}X^2 = M^{*}[Z_1^2/U_1^2,\ldots,Z_p^2/U_p^2]X^2 = M^{*}X^2
\]
and
\begin{eqnarray*}
    M^{Z^2,*}X^2 & = & ((M^{Z^2} X^2 \cap \pow(D^*)) \setminus \{M^{Z^2,*}X^1 : X^1 \in ({\cal V}'_{1}\cup {\cal V}_{1}^{F})\})\\
	     &   &  \qquad\cup \{M^{Z^2,*}X^1 : X^1 \in ({\cal V}'_{1}\cup {\cal V}_{1}^{F}), M^{Z^2}X^1 \in M^{Z^2}X^2\} \\
	     & = & ((MX^2 \cap pow(D^*)) \setminus \{M^{*}X^1 : X^1 \in ({\cal V}'_{1}\cup {\cal V}_{1}^{F})\}) \\
	     &   &  \qquad\cup \{M^{*}X^1 : X^1 \in ({\cal V}'_{1}\cup {\cal V}_{1}^{F}), MX^1 \in MX^2\}\\
	     & = & M^{*}X^2 \,.
\end{eqnarray*}
Since $M^{*,Z^2}X^2 = M^{Z^2,*}X^2$ the thesis follows, at least in
the case in which $X^2 \notin \{Z_1^2,\ldots,Z_p^2\}$.
On the other
hand, if $X^2 \in \{Z_1^2,\ldots,Z_p^2\}$, say $X^2 = Z_j^2$, then
$M^{*,Z^2}X^2 = U_{j}^2$, and
\begin{eqnarray*}
    M^{Z^2,*}X^2 & = & ((M^{Z^2} X^2 \cap \pow(D^*)) \setminus \{M^{Z^2,*}X^1 : X^1 \in ({\cal V}'_{1}\cup {\cal V}_{1}^{F})\})\\
	     &   &  \qquad\cup \{M^{Z^2,*}X^1 : X^1 \in ({\cal V}'_{1}\cup {\cal V}_{1}^{F}), M^{Z^2}X^1 \in M^{Z^2}X^2\} \\
	     & = & (U_{j}^2  \setminus \{M^{*}X^1 : X^1 \in ({\cal V}'_{1}\cup {\cal V}_{1}^{F})\}) \\
	     &   &  \qquad\cup (\{M^{*}X^1 : X^1 \in ({\cal V}'_{1}\cup {\cal V}_{1}^{F}), MX^1 \in U_{j}^2\} )\\
	     & = & U_{j}^2 \,.
\end{eqnarray*}
Clearly the thesis follows also in this case.

\item  Let $X^3 \in {\cal V}_3$. Then we have
    \hspace*{-2cm}
    \begin{eqnarray}
        M^{*,Z^2}X^3 & \!\!\!=\!\!\! & M^{*}X^3 =
                        ((MX^3 \cap \pow(\pow(D^{*}))) \setminus
                        \{M^{*}X^2: X^2 \in {\cal V}'_{2}\}) \notag\\
                 &   &  \phantom{M^{*}A =} \qquad\cup \{M^{*}X^2: X^2 \in {\cal V}'_{2},~MX^2 \in MX^3\}
                       \label{a5}\\
        M^{Z^2,*}X^3 & \!\!\!=\!\!\! & ((M^{Z^2}X^3 \cap \pow(\pow(D^{*}))) \setminus
                        \{M^{Z^2,*}X^2: X^2 \in {\cal V}'_{2} \cup
                             \{Z_{1}^2,\ldots,Z_{p}^2\}\}) \notag\\
                 &   &  \qquad {}\cup \{M^{Z^2,*}X^2: X^2 \in {\cal V}'_{2} \cup
                             \{Z_{1}^2,\ldots,Z_{p}^2\},~M^{Z^2}X^2 \in M^{Z^2}X^3\} \notag\\
                 & \!\!\!=\!\!\! & ((MX^3 \cap \pow(\pow(D^{*}))) \setminus
                        (\{M^{*}X^2: X^2 \in {\cal V}'_{2}\} \cup
                        \{U_{j}^2: j = 1,\ldots,p \})) \notag\\
                 &   &  \qquad {}\cup \{M^{*}X^2: X^2 \in {\cal
		 V}'_{2},~MX^2 \in MX^3\} \notag\\
                 &   &  \qquad {}\cup (\{U_{j}^2: j = 1,\ldots,p\} \cap
			MX^3).\label{a6}
    \end{eqnarray}

    By putting
    \[
    \begin{array}{lll}
        P_{1} & = & MX^3 \cap \pow(\pow(D^{*}))\,,  \\
        P_{2} & = & \{M^{*}X^2: X^2 \in {\cal V}'_{2}\}\,,  \\
        P_{3} & = & \{U_{j}^2: j = 1,\ldots,p \}\,,  \\
        P_{4} & = & \{M^{*}X^2: X^2 \in {\cal V}'_{2},~MX^2 \in MX^3\}\,,  \\
        P_{5} & = & \{U_{j}^2: j = 1,\ldots,p\} \cap MX^3\,,
    \end{array}
    \]
    then (\ref{a5}) and (\ref{a6})
    can be respectively rewritten as
    \begin{eqnarray}
        M^{*,Z^2}X^3 & = & (P_{1} \setminus P_{2}) \cup P_{4} \label{a7}\\
        M^{Z^2,*}X^3 & = & (P_{1} \setminus (P_{2} \cup P_{3}))
                       \cup P_{4} \cup P_{5}\,.\label{a8}
    \end{eqnarray}
    Moreover, it is easy to verify that the following relations hold:
    \[
        P_{2} \cap P_{3}  =  \emptyset\,,~~~
        P_{5}  =  P_{1} \cap P_{3}\,,~~~ \text{and~~~}
        P_{4}  \subseteq  P_{2}\,,
    \]
    so that
    \begin{eqnarray}
    (P_{1} \setminus P_{2}) \cup P_{4} &=&
    (P_{1} \setminus (P_{2} \cup P_{3}))
     \cup P_{4} \cup (P_{1} \cap P_{3}) \notag\\
    &=& (P_{1} \setminus (P_{2} \cup P_{3})) \cup P_{4} \cup P_{5}\,.
    \label{a9}
    \end{eqnarray}
    Therefore, in view of (\ref{a7}) and (\ref{a8}) above,
    (\ref{a9})  yields $M^{*,Z^2}X^3 = M^{Z^2,*}X^3$.
\end{enumerate}
\end{proof}

The following lemma proves that satisfiability is preserved in the case of purely universal formulae.
\begin{lemma}\label{quantifiedform}
Let $(\forall z_1) \ldots (\forall z_n) \varphi_0$, $(\forall Z_1^1)
\ldots (\forall Z_m^1) \varphi_1$, and $(\forall Z_1^2) \ldots
(\forall Z_p^2) \varphi_2$ be conjuncts of $\psi$.  Then
\begin{itemize}
\item [(i)] if $\model\models (\forall z_1) \ldots (\forall z_n)
\varphi_0$, then $\model^*\models (\forall z_1) \ldots (\forall z_n)
\varphi_0$;

\item [(ii)] if $\model\models (\forall Z_1) \ldots (\forall Z_m)
\varphi_1$, then $\model^*\models (\forall Z_1) \ldots (\forall Z_m)
\varphi_1$;
\item [(iii)] if $\model\models (\forall Z_1^2) \ldots (\forall Z_p^2)
\varphi_2$, then $\model^*\models (\forall Z_1^2) \ldots (\forall Z_p^2)
\varphi_2$.
\end{itemize}
\end{lemma}
\begin{proof}
\begin{itemize}
\item [(i)]
Assume by contradiction that there exist $u_1,\ldots,u_n \in D^*$
such that $\model^{*,z}\not\models \varphi_0$.
Then, there must be an atomic formula $\varphi_0'$ in $\varphi_0$ that
is interpreted differently in $\model^{*,z}$ and in $\model^z$.
Recalling that $\varphi_0$ is a propositional combination of
quantifier-free atomic formulae of any level, let us first suppose
that $\varphi_0'$ is $X^2 = Y^2$ and, without loss of generality,
assume that $\model^{*,z}\not \models X^2 = Y^2$.  Then $M^{*,z}X^2
\neq M^{*,z}Y^2$, so that, by Lemma \ref{le_M*zMz*}, $M^{z,*}X^2 \neq
M^{z,*}Y^2$.  Then, Lemma \ref{le_basic} yields $M^zX^2 \neq M^zY^2$,
a contradiction.  The other cases are proved in an analogous way.

\item [(ii)] This case can be proved much along the same lines as the
proof of case (ii) of Lemma 4 in \cite{CanNic08}.  Here, one has to take
care of the fact that $\varphi_1$ may contain purely universal
formulae of level 1 occurring only positively in $\varphi_1$ and not
satisfying Restriction \ref{restriction1} of Section
\ref{restrictionquant}.  This is handled similarly to case (i)
of this lemma.  Another issue that has to be considered is the fact
that the collection of relevant variables of sort 1 for $\psi$ are not
just the variables occurring free in $\psi$, namely the ones in ${\cal
V}_1'$, but also the variables in ${\cal V}_1^{F}$, introduced to
denote the elements distinguishing the sets $M^*X^{2}$, for $X^{2} \in
\mathcal{V}'_{2}$.

\item [(iii)]
Assume, by way of contradiction, that
$\model \models(\forall
Z_1^2) \ldots (\forall Z_p^2) \varphi_2$,
but
$\model^* \not\models(\forall
Z_1^2) \ldots (\forall Z_p^2) \varphi_2$. Hence there exist $U_1^2,
\ldots, U_p^2 \in \pow(\pow(D^*))$ such that $\model^{*,Z^2}\not\models\varphi_2$.

Without loss of generality, assume that $U_i^2 = M^*X_i^2$, for $1
\leq i \leq k$ and where $X_1^2,\ldots,X_k^2 \in {\cal V}_2'$, and
that $U_j^2 \neq M^*X^2$, for $k+1 \leq j \leq p$ and
$X^2 \in {\cal V}_2'$, for some $k \geq 0$.

Let $\bar{\varphi}_2$ be the formula obtained by simultaneously
substituting $Z_1^2,\ldots,Z_k^2$ with $X_1^2,\ldots,X_k^2$ in
$\varphi_2$, and let $\model^{*,Z_k^2} =
\model^*[Z_{k+1}^2/U_{k+1}^2,\ldots,Z_{p}^2/U_p^2]$.  Further, let
${\cal M}^{Z^{2'}}$ be a $\QLQSR$-interpretation differing from
$\model^{Z^2}$ only in the evaluation of $Z_1^2,\ldots,Z_k^2$, with
$M^{Z^{2'}}Z_1^2 = MX_1^2,\ldots, M^{Z^{2'}}Z_k^2 = MX_k^2$.

We distinguish the following two cases:

\begin{description}
    \item[Case $k = p$:]
If $k = p$, then $\model^{*,Z_k^2}$ and $\model^*$ coincide and a
contradiction can be obtained by showing that the implications
\[
\model^{*,Z^2} \not\models \varphi_2 ~\Rightarrow~ \model^*\not\models
\bar{\varphi}_2 ~\Rightarrow~ \model \not\models \bar{\varphi}_2
~\Rightarrow~ \model^{Z^{2'}}\not\models \varphi_2
\]
hold, since these together with the fact that $\model^{*,Z^2}
\not\models \varphi_2$ would yield $\model \not\models (\forall
Z_1^2)\ldots(\forall Z_p^2)\varphi_2$, contradicting our initial
hypothesis.  The first implication,
$\model^{*,Z^2} \not\models \varphi_2 \Rightarrow \model^*\not\models
\bar{\varphi}_2$, is plainly derived from the definition of
$\bar{\varphi}_2$.  The second one, $\model^*\not\models
\bar{\varphi}_2 \Rightarrow \model \not\models \bar{\varphi}_2$, can
be proved as follows.  For every purely universal formula either of
level 1 or of level 2, $\bar{\varphi}_2'$, occurring only positively
in $\bar{\varphi}_2$, it follows that $\model^*\not\models
\bar{\varphi}_2' \Rightarrow \model \not\models \bar{\varphi}_2'$ by
reasoning as in case (i) or in case (ii) of the present lemma, respectively.  For
each other atomic formula $\bar{\varphi}_2'$ occurring in
$\bar{\varphi}_2$ we have to show that $\model^*$ and $\model$
evaluate $\bar{\varphi}_2'$ in the same manner.  If
$\bar{\varphi}_2'$ is a quantifier-free atomic formula, the proof
follows directly from Lemma \ref{le_basic}.  If $\bar{\varphi}_2'$ is
an atomic formula of level 1, it can only be of type $(\forall
z_1)\ldots(\forall z_n)\neg(\bigwedge_{i,j=1}^n \langle
z_i,z_j\rangle=Y_{ij}^2)$, where $Y_{ij}^2$ is any variable in
${\cal V}_2$.
Reasoning analogously to case (i) of the present lemma, it follows that $\model \models
\bar{\varphi}_2'\Rightarrow \model^* \models \bar{\varphi}_2'$.  Next,
let us prove that $\model^* \models \bar{\varphi}_2'\Rightarrow \model
\models \bar{\varphi}_2'$.  Assume by contradiction that $\model
\not\models \bar{\varphi}_2'$.  That is, $\model \not\models (\forall
z_1)\ldots(\forall z_n)\neg(\bigwedge_{i,j=1}^n\langle
z_i,z_j\rangle=Y_{ij}^2)$.  Then, there are $u_1,\ldots,u_n \in D$
such that $\model[z_1/u_1,\ldots,z_n/u_n]\models \bigwedge_{i,j=1}^n\langle z_i,z_j\rangle=Y_{ij}^2$.  By the
construction in Section \ref{decisionproc}, all these $u_i$s are in
$D^*$, $M Y_{ij} = M^* Y_{ij}$ and thus we finally obtain that
\[
\model^* \not\models (\forall z_1)\ldots(\forall
z_n)\neg(\bigwedge_{i,j=1}^n\langle
z_i,z_j\rangle=Y_{ij}^2),
\]
contradicting our hypothesis.

Finally, $\model \not\models \bar{\varphi}_2 \Rightarrow
\model^{Z^{2'}}\not\models \varphi_2$, follows from the definition of
$\bar{\varphi}_2$ and of $Z^{2'}$.


\item[Case $k < p$:] In this case, the schema of the proof is
analogous to the one in the previous case.  However, since
$\model^{*,Z_k^2}$ and $\model^*$ do not coincide, the single steps
are carried out in a slightly different manner.  Thus, for the sake of
clarity we report below the details of the proof.

In order to obtain a contradiction we prove that the following
implications hold
\[
\model^{*,Z^2} \not\models \varphi_2 ~\Rightarrow~
\model^{*,Z_k^2}\not\models \bar{\varphi}_2 ~\Rightarrow~
\model^{Z_k^2}\not\models \bar{\varphi}_2 ~\Rightarrow~
\model^{Z^{2'}}\not\models \varphi_2\,.
\]

The first implication, $\model^{*,Z^2} \not\models \varphi_2
\Rightarrow \model^{*,Z_k^2} \not\models \bar{\varphi}_2$, can be
immediately deduced from the definition of $\bar{\varphi}_2$ and of
$\model^{*,Z_k^2}$.  The second implication,
$\model^{*,Z_k^2}\not\models \bar{\varphi}_2 \Rightarrow
\model^{Z_k^2}\not\models \bar{\varphi}_2$, can be proved as shown
next.  If $\bar{\varphi}_2'$ is a purely universal formula either of
level 1 or of level 2 occurring only positively in $\bar{\varphi}_2$,
we have $\model^{*,Z_k^2}\not\models \bar{\varphi}_2'$ and,
since $\model^{*,Z_k^2}$ and
$\model^{Z_k^2,*}$ coincide (by Lemma \ref{le_eqM*ZMZ*2}), we obtain
$\model^{Z_k^2,*}\not\models \bar{\varphi}_2'$.  Then, reasoning as in case
(i) (if $\bar{\varphi}_2'$ is of level 1) or in case (ii) (if
$\bar{\varphi}_2'$ is of level 2) of the present lemma, it follows
that $\model^{Z_k^2}\not\models \bar{\varphi}_2'$.  If
$\bar{\varphi}_2'$ is a quantifier-free atomic formula occurring in
$\bar{\varphi}_2$, we prove that $\bar{\varphi}_2'$ in
$\bar{\varphi}_2$ is interpreted in $\model^{*,Z_k^2}$ and in
$\model^{Z_k^2}$ in the same way, using Lemmas \ref{le_eqM*ZMZ*2} and
\ref{le_basic}.

If $\bar{\varphi}_2'$ is a purely universal formula of level 1, it
must have the form
\[
(\forall z_1)\ldots(\forall
z_n)\neg(\bigwedge_{i,j=1}^n\langle
z_i,z_j\rangle=Y_{ij}^2)\,,
\]
where $Y_{ij}^2$ is any variable in ${\cal V}_2$.  In this case the
proof is carried out as shown next.  Reasoning as in case (i), we have $\model^{Z_k^2} \models
\bar{\varphi}_2'\Rightarrow \model^{Z_k^2,*} \models
\bar{\varphi}_2'$, and by Lemma \ref{le_eqM*ZMZ*2}, that
$\model^{*,Z_k^2} \models \bar{\varphi}_2'$. Proceeding as in the first case of this item of
the present lemma,
we obtain that
$\model^{Z_k^2,*} \not\models (\forall z_1)\ldots(\forall
z_n)\neg(\bigwedge_{i,j=1}^n\langle
z_i,z_j\rangle=Y_{ij}^2)$ and, by Lemma \ref{le_eqM*ZMZ*2}, that
$\model^{*,Z_k^2} \not\models (\forall z_1)\ldots(\forall
z_n)\neg(\bigwedge_{i,j=1}^n\langle
z_i,z_j\rangle=Y_{ij}^2)$, contradicting our hypothesis.

Finally, the third implication, $\model^{Z_k^2}\not\models
\bar{\varphi}_2 \Rightarrow \model^{Z^{2'}}\not\models \varphi_2$
follows directly from the definition of $\bar{\varphi}_2$ and of
$Z^{2'}$.
\end{description}
\end{itemize}
\end{proof}
We can now state and prove our main result.
\begin{theorem}\label{correctness}
Let $\model$ be a $\QLQSR$-interpretation satisfying a normalized
$\QLQSR$-conjunction $\psi$.  Then $\model^*\models \psi$, where
$\model^*$ is the relativized interpretation of $\model$ with respect
to a domain $D^{*}$ satisfying (\ref{DStar}).
\end{theorem}
\begin{proof}
We only have to prove that $\model^*\models \psi'$, for each conjunct
$\psi'$ occurring in $\psi$.  Each such $\psi'$ must be of one of the
types (1)--(3) enumerated in Section \ref{normal3LQS}.  By applying
either Lemma \ref{le_basic} or Lemma \ref{quantifiedform} to each
$\psi'$ (according to its type) we obtain the thesis.
\end{proof}
From the above reduction and relativization steps, the following
result follows easily:
\begin{corollary}
    The fragment $\QLQSR$ enjoys a small model property (and
    therefore it has a solvable satisfiability problem). 
\end{corollary}

\section{Expressiveness of the language $\QLQSR$}\label{sec:applications}
\label{expressiveness}
Much as shown in \cite{CanNic08}, the language $\QLQSR$ can express a
restricted variant of the set-formation operator, which in turn allows one to
express other significant set operators such as binary union,
intersection, set difference, the singleton operator, the powerset
operator (over subsets of the universe only), etc.  More specifically,
atomic formulae of type $X^{i}=\{X^{i-1} :
\varphi(X^{i-1})\}$, for $i = 1,2,3$, can be expressed in $\QLQSR$ by
the formulae
\[
(\forall X^{i-1})(X^{i-1} \in X^{i} \leftrightarrow
\varphi(X^{i-1}))
\]
provided that the syntactic constraints of $\QLQSR$ are satisfied.

Since $\QLQSR$ is a superlanguage of $\mathit{3LQS^{R}}$, the language
$\QLQSR$ can express the syllogistic $2LS$ (cf.\ \cite{FerOm1978}) and
the sublanguage $3LSSP$ of $3LSSPU$ not involving the set-theoretic
construct of general union, since these are expressible in
$\mathit{3LQS^{R}}$, as shown in \cite{CanNic08}.  We recall that
$3LSSPU$ admits variables of three sorts and, besides the usual
set-theoretical constructs, it involves the `singleton set' operator
$\{\cdot\}$, the powerset operator $\pow$, and the general union
operator $\mathit{Un}$.
$3LSSP$ can plainly be decided by the decision procedure presented in
\cite{CanCut93} for the whole fragment $3LSSPU$.

Among the other constructs of set theory which are expressible in the
language $\QLQSR$ (cf.\ \cite{CanNic08}), we cite:
\begin{itemize}

\item literals of the form $X^2 = \pow_{< h}(X^1)$, where $\pow_{<
h}(X^1)$ denotes the collection of subsets of $X^1$ with less than $h$
elements;

\item the unordered Cartesian product $X^2 = X_1^1 \otimes \ldots
\otimes X_n^1$, where $X_1^1 \otimes \ldots
\otimes X_n^1$ denotes the collection $\{\{ x_{1},\ldots,x_{n}\} : x_{1} \in X_1^1,
\ldots, x_{n} \in X_n^1\}$;

\item literals of the form $A = \pow^*(X_1^1,\ldots,X_n^1)$, where
$\pow^*(X_1^1,\ldots,X_n^1)$ is the variant of the powerset
introduced in \cite{Can91} which denotes the collection
$$
\{Z : Z \subseteq \bigcup_{i=1}^n X_i^1
         \hbox{ and } Z \cap X_i^1 \neq \emptyset,
         \hbox{ for all } 1 \leq i \leq n \},
$$
\end{itemize}
For instance, a literal of the form $X^2 = \pow_{< h}(X^1)$, with $h
\geq 2$, can be expressed by the $\QLQSR$-formula
\[
(\forall Y^{1})\Bigg(Y^{1} \in X^2 \leftrightarrow \Bigg( (\forall
z)\left(z
\in Y^{1} \rightarrow z \in X^1\right) \And
(\forall z_{1})\ldots(\forall z_{h})\Bigg(\bigwedge_{i=1}^{h} z_{i}
\in Y^1 \rightarrow \bigvee_{\substack{i,j=1\\i<j}}^{h} z_{i} = z_{j}\Bigg)
\!\Bigg)\!\Bigg),
\]
as can be easily verified.

\subsection{Other applications of $\QLQSR$}
Within the $\QLQSR$ language it is also possible to define binary
relations over elements of a domain together with several conditions
on them which characterize accessibility relations of well-known modal
logics.  These formalizations are illustrated in Table
\ref{tab:accesrel}.

\begin{table}[tb]
\begin{tabular}{|l|l|}
  \hline
  Binary relation & $(\forall Z^2)(Z^2 \in X_R^3 \leftrightarrow \neg(\forall z_1, z_2)\neg(\langle z_1,z_2\rangle = Z^2))$\\\hline\hline
  Reflexive & $(\forall z_1)(\langle z_1,z_1\rangle \in X_R^3)$\\\hline
  Symmetric & $(\forall z_1,z_2)(\langle z_1,z_2\rangle \in X_R^3 \rightarrow \langle z_2,z_1\rangle \in X_R^3)$\\\hline
  Transitive & $(\forall z_1,z_2,z_3)((\langle z_1,z_2\rangle \in X_R^3 \And \langle z_2,z_3\rangle \in X_R^3)\rightarrow \langle z_1,z_3\rangle \in X_R^3)$\\\hline
  Euclidean  & $(\forall z_1,z_2,z_3)((\langle z_1,z_2\rangle \in X_R^3 \And \langle z_1,z_3\rangle \in X_R^3)\rightarrow \langle z_2,z_3\rangle \in X_R^3)$ \\ \hline
  Weakly-connected &  $(\forall z_1,z_2,z_3)((\langle z_1,z_2\rangle \in X_R^3 \And \langle z_1,z_3\rangle \in X_R^3)$ \\
                   & \hfill $ \rightarrow (\langle z_2,z_3\rangle \in X_R^3 \vee z_2 = z_3 \vee \langle z_3,z_2\rangle \in X_R^3))$\\ \hline
  Irreflexive & $(\forall z_1)\neg (\langle z_1,z_1\rangle\in X_R^3)$\\ \hline
  Intransitive & $(\forall z_1,z_2,z_3)((\langle z_1,z_2\rangle \in X_R^3 \And \langle z_2,z_3\rangle \in X_R^3)\rightarrow \neg \langle z_1,z_3\rangle \in X_R^3)$\\\hline
  Antisymmetric & $(\forall z_1,z_2)((\langle z_1,z_2\rangle \in X_R^3 \wedge \langle z_2,z_1\rangle \in X_R^3)\rightarrow (z_1 = z_2))$\\\hline
  Asymmetric & $(\forall z_1,z_2)(\langle z_1,z_2\rangle \in X_R^3 \rightarrow \neg(\langle z_2,z_1\rangle \in X_R^3))$\\
  \hline
\end{tabular}
\caption{\label{tab:accesrel} $\QLQSR$ formalization of conditions of accessibility relations}
\end{table}

Usual Boolean operations over relations can be defined as shown in
Table \ref{tab:Bolop}.
\begin{table}[tb]
\begin{center}
\begin{tabular}{|l|l|l|}
  \hline
Intersection & $X_R^3 = X_{R_1}^3 \cap X_{R_2}^3$ & $(\forall Z^2)(Z^2 \in X_R^3 \leftrightarrow (Z^2 \in X_{R_1}^3 \And Z^2 \in X_{R_2}^3))$\\\hline
Union & $X_R^3 = X_{R_1}^3 \cup X_{R_2}^3$ & $(\forall Z^2)(Z^2 \in X_R^3 \leftrightarrow (Z^2 \in X_{R_1}^3 \Or Z^2 \in X_{R_2}^3))$\\\hline
Complement & $X_{R_1}^3 = \overline{X_{R_2}^3}$ & $(\forall Z^2)(Z^2 \in X_{R_1}^3 \leftrightarrow \neg(Z^2 \in \overline{X_{R_2}^3}))$ \\\hline
Set difference & $X_R^3 = X_{R_1}^3 \setminus X_{R_2}^3$ & $(\forall Z^2)(Z^2 \in X_R^3 \leftrightarrow (Z^2 \in X_{R_1}^3 \And \neg (Z^2 \in X_{R_2}^3)))$\\\hline
Set inclusion &  $X_{R_1}^3 \subseteq X_{R_2}^3$ & $(\forall Z^2)(Z^2 \in X_{R_1}^3 \rightarrow Z^2 \in X_{R_2}^3)$\\
 \hline
\end{tabular}
\caption{\label{tab:Bolop} $\QLQSR$ formalization of Boolean operations over relations}
\end{center}
\end{table}
The language $\QLQSR$ allows one also to express the inverse
$X_{R_2}^3$ of a given binary relation $X_{R_1}^3$ (namely, to express
the literal $X_{R_2}^3 = (X_{R_1}^3)^{-1}$) by means of the
$\QLQSR$-formula $(\forall z_1,z_2)\left(\langle z_1,z_2\rangle \in
X_{R_1}^3 \leftrightarrow \langle z_2,z_1 \rangle \in X_{R_2}^3\right)$.

In the next section we will present an application of the decision
procedure for $\QLQSR$-formulae to modal logic.
For this purpose we introduce below a family $\{(\QLQSR)^h\}_{h \geq
2}$ of fragments of $\QLQSR$, each of which has an
\textsf{NP}-complete satisfiability problem, and then show, in the
next section, that the modal logic $\Kqc$ can be formalized in
$(\QLQSR)^3$ in a succint way, thus rediscovering the
\textsf{NP}-completeness of the decision problem for $\Kqc$ (cf.\
\cite{Lad77}).

Formulae in $(\QLQSR)^h$ must satisfy various syntactic constraints.
First of all, all quantifier prefixes occurring in a formula in
$(\QLQSR)^h$ must have their length bounded by the constant $h$.
Thus, given a satisfiable $(\QLQSR)^h$-formula $\varphi$ and a
$\QLQSR$-model $\model = (D, M)$ for it, from
Theorem~\ref{correctness} it follows that $\varphi$ is satisfied by
the relativized interpretation $\model^* = (D^{*}, M^{*})$ of $\model$
with respect to a domain $D^{*}$ whose size is bounded by the
expression in (\ref{DStar}).  But since in this case $L_{m} \leq h$
and $L_{n} \leq h$, where $L_{m}$ and $L_{n}$ are defined as in Step 4
of the construction of $D^{*}$ (cf.\ Section~\ref{ssseUniv}), it
follows that the bound in (\ref{DStar}) is quadratic in the size of
$\varphi$.  The remaining syntactic constraints on
$(\QLQSR)^h$-formulae will allow us to deduce that $M^{*}X^2 \subseteq
\pow_{< h}(D^{*})$, for any free variable $X^{2}$ of sort 2 in
$\varphi$, and $M^{*}X^3 \subseteq \pow_{< h}(\pow_{< h}(D^{*}))$, for
any free variable $X^{3}$ of sort 3 in $\varphi$, so that the model
$\model^{*}$ can be guessed in nondeterministic polynomial time in the
size of $\varphi$, and one can check in deterministic polynomial time
that $\model^{*}$ actually satisfies $\varphi$, proving that the
satisfiability problem for $(\QLQSR)^h$-formulae is in \textsf{NP}.
As the satisfiability problem \textsf{SAT} for propositional logic can
be readily reduced to that for $(\QLQSR)^h$-formulae, the
\textsf{NP}-completeness of the latter problem follows.

\begin{definition}[$(\QLQSR)^h$-formulae]\label{def:hlang}
Let $\varphi$ be a $\QLQSR$-formula involving the designated free
variables $X_U^1$, $X_{< h}^2$, and $X_{< h}^3$ (of sort 1, 2, and 3,
respectively).  Let $X_1^2,\ldots,X_p^2$ be the free variables of sort
2 occurring in $\varphi$, distinct from $X_{< h}^2$.  Likewise, let
$X_1^3,\ldots,X_k^3$ be the free variables of sort 3 occurring in
$\varphi$, distinct from $X_{< h}^3$. Then $\varphi$ is a
$(\QLQSR)^h$-formula, with $h \geq 2$, if it has the form (up to the
order of the conjuncts)
\[
\xi_U^1 \And \xi_{< h}^2 \And \xi_{< h}^3 \And \psi_1^2 \And
\ldots \And \psi_p^2 \And \psi_1^3\And \ldots\And \psi_k^3 \And
\chi \, ,
\]
where
\begin{enumerate}
\item $\xi_U^1 \equiv (\forall z)(z \in X_U^1)$,

\emph{i.e., $X_U^1$ is the (nonempty) universe of discourse};

\item $\xi_{< h}^2 \equiv (\forall Z^1)\left(Z^1 \in X_{< h}^2
\rightarrow
      (\forall z_1)\ldots (\forall
      z_{h})\left(\bigwedge_{i=1}^{h} z_i \in Z^1 \rightarrow
      \bigvee_{i,j=1, i<j}^{h} z_i = z_j\right)\right)$,

\emph{i.e., $X_{< h}^2 \subseteq \pow_{< h}(X_U^1)$ (together with
formula $\xi_U^1$)};

\item $\xi_{< h}^3 \equiv (\forall Z^2)\Big(Z^2 \in X_{< h}^3
\rightarrow \Big((\forall Z^1)(Z^1 \in Z^2 \rightarrow Z^1 \in
X_{< h}^2)$

      $\phantom{\xi_{\leq h}^2 =}\And(\forall Z_1^1)\ldots
  (\forall Z_{h}^1)\big(\bigwedge_{i=1}^{h} Z_i^1 \in Z^2 \rightarrow
      \bigvee_{i,j=1, i<j}^{h} Z_i^1 = Z_j^1\big)\Big)\Big)$,

\emph{i.e., $X_{< h}^3 \subseteq \pow_{< h}(\pow_{< h}(X_U^1))$
(together with formulae $\xi_U^1$ and $\xi_{< h}^2$)};

\item either $\psi_i^2 \equiv (\forall Z^1)(Z^1 \in X_i^2 \rightarrow
Z^1 \in X_{< h}^2)$ or $\psi_i^2 \equiv X_i^2 \in X_{< h}^3$, for $i
= 1,\ldots,p$,

\emph{so that, $X_i^2 \subseteq \pow_{< h}(X_U^1)$, for $i = 1,\ldots,p$
(together with formulae $\xi_U^1$ and $\xi_{< h}^2$)};

\item $\psi_j^3 \equiv (\forall Z^2)(Z^2 \in
X_j^3 \rightarrow Z^2 \in X_{< h}^3)$, for $j = 1,\ldots,k$,

\emph{i.e., $X_j^3 \subseteq \pow_{< h}(\pow_{< h}(X_U^1))$, for $j =
1,\ldots,k$ (together with formulae $\xi_U^1$, $\xi_{< h}^2$, and
$\xi_{< h}^3$)};

\item $\chi$ is a propositional combination of
\begin{enumerate}
    \item quantifier-free atomic formulae of any level,

    \item\label{levelOne} purely universal formulae of level 1 of the form
    \[
    (\forall
    z_1)\ldots(\forall z_n)\varphi_0\,,
    \]
    with $n \leq h$,

    \item\label{levelTwo} purely universal formulae of level 2 of the form
    \[
    (\forall
    Z_1^1)\ldots (\forall Z_m^1)((Z_1^1 \in X_{< h}^2 \wedge \ldots
    \wedge Z_m^1 \in X_{< h}^2) \rightarrow \varphi_1)\,,
    \]
    where $m \leq h$ and $\varphi_1$ is a propositional combination of
    quantifier-free atomic formulae and of purely universal formulae
    of level 1 satisfying (\ref{levelOne}) above,

    \item purely universal formulae of level 3 of the form
    \[
    (\forall
    Z_1^2)\ldots (\forall Z_p^2)((Z_1^2 \in X_{< h}^3 \wedge \ldots
    \wedge Z_p^2 \in X_{< h}^3) \rightarrow \varphi_2)\,,
    \]
    where $p \leq h$ and $\varphi_2$ is a propositional combination of
    quantifier-free atomic formulae, and of purely universal formulae
    of level 1 and of level 2 satisfying (\ref{levelOne}) and
    (\ref{levelTwo}) above.
\end{enumerate}
\end{enumerate}

\end{definition}

Having defined the fragments $(\QLQSR)^{h}$, for $h \geq 2$, next we
prove that each of them has an \textsf{NP}-complete satisfiability
problem.
\begin{theorem}
The satisfiability problem for $(\QLQSR)^{h}$ is
\textsf{NP}-complete, for any $h \geq 2$.
\end{theorem}
\begin{proof}
The satisfiability problem \textsf{SAT} for propositional logic can be
readily reduced to the one for $(\QLQSR)^{h}$-formulae, for any $h
\geq 2$, as follows. Given a formula $\rho \in \textsf{SAT}$, we
construct a quantifier-free $(\QLQSR)^{h}$-formula $\varphi_{\rho}$
by replacing each propositional letter $P_{i}$ in $\rho$ by the
quantifier-free formula $x_{i} \in X^{1}$, where $X^{1}$ is a fixed
variable of sort 1 and the $x_{i}$s are distinct variables of sort
0 in a one-one correspondence with the distinct propositional
letters in $\rho$. Plainly, $\rho$ is propositionally satisfiable if
and only if $\varphi_{\rho}$ is satisfiable by a $\QLQSR$-model.
Therefore the \textsf{NP}-hardness of the satisfiability problem for
$(\QLQSR)^{h}$-formulae follows.

To prove that our problem is in \textsf{NP}, we reason as follows.
Let
\begin{equation}
    \label{eq_QLQSR}
    \varphi \equiv \xi_U^1 \And \xi_{< h}^2 \And \xi_{< h}^3 \And \psi_1^1
    \And \ldots \And \psi_p^1 \And \psi_1^2\And \ldots\And \psi_k^2
    \And \chi
\end{equation}
be a satisfiable $(\QLQSR)^{h}$-formula, and let $\Hf$ be a set of
formulae constructed as follows.  Initially, we put
\[
\Hf := \{\xi_U^1, \xi_{< h}^2, \xi_{< h}^3,\psi_1^1,\ldots, \psi_p^1,
\psi_1^2, \ldots, \psi_k^2, \chi\}
\]
and then, we modify $\Hf$ according to the following six rules, until
no rule can be further applied:\footnote{We recall that an
implication $A \rightarrow B$ has to be regarded as a shorthand for
the disjunction $\neg A \vee B$.}
\begin{enumerate}[~~R1:]
\item if $\xi \equiv \neg\neg \xi_1$ is in $\Hf$, then $\Hf = (\Hf
\setminus \{\xi\}) \cup \{\xi_1\}$,

\item if $\xi \equiv \xi_1 \And \xi_2$ (resp., $\xi \equiv \neg (\xi_1
\Or \xi_2)$) is in $\Hf$ (i.e., $\xi$ is a conjunctive formula), then
we put $\Hf := (\Hf \setminus \{\xi\}) \cup \{\xi_1,\xi_2\}$
(resp., $\Hf := (\Hf \setminus \{\xi\}) \cup \{\neg\xi_1,\neg\xi_2\}$),

\item if $\xi \equiv \xi_1 \Or \xi_2$ (resp., $\xi \equiv \neg (\xi_1
\And \xi_2)$) is in $\Hf$ (i.e., $\xi$ is a disjunctive formula), then
we choose a $\xi_i$, $i \in \{1,2\}$, such that $H_{\varphi}\cup
\{\xi_i\}$ (resp., $H_{\varphi}\cup
\{\neg\xi_i\}$)
is satisfiable and put $\Hf := (\Hf \setminus \{\xi\}) \cup
\{\xi_i\}$ (resp., $\Hf := (\Hf \setminus \{\xi\}) \cup
\{\neg\xi_i\}$),

\item if $\xi \equiv \neg(\forall z_1)\ldots (\forall z_n)\varphi_0$
is in $\Hf$, then $\Hf := (\Hf \setminus \{\xi\}) \cup
\{\neg(\varphi_0)_{\bar{z}_1,\ldots,\bar{z}_n}^{z_1,\ldots,z_n}\}$,
where $\bar{z}_1,\ldots,\bar{z}_n$ are newly introduced variables of
sort 0,

\item if $\xi \equiv \neg(\forall Z_1^1)\ldots (\forall Z_m^1)\varphi_1$
is in $\Hf$, then $\Hf := (\Hf \setminus \{\xi\}) \cup
\{\neg(\varphi_1)_{\bar{Z}_1^1,\ldots,\bar{Z}_m^1}^{Z_1^1,\ldots,Z_m^1}\}$,
where $\bar{Z}_1^1,\ldots,\bar{Z}_m^1$ are fresh variables of sort 1,

\item if $\xi \equiv \neg(\forall Z_1^2)\ldots (\forall
Z_p^2)\varphi_2$ is in $\Hf$, then $\Hf := (\Hf \setminus \{\xi\})
\cup
\{\neg(\varphi_2)_{\bar{Z}_1^2,\ldots,\bar{Z}_p^2}^{Z_1^2,\ldots,Z_p^2}\}$,
where $\bar{Z}_1^2,\ldots,\bar{Z}_p^2$ are newly introduced variables
of sort 2.
\end{enumerate}
Plainly, the above construction terminates in $\mathcal{O}(|\varphi|)$
steps and if we put $\psi \equiv \bigwedge_{\xi \in \Hf}\xi$, it turns
out that
\begin{enumerate}[~~~~(a)]
\item $\psi$ is a satisfiable $(\QLQSR)^h$-formula,

\item $|\psi| = \mathcal{O}(|\varphi|)$, and

\item $\psi \rightarrow \varphi$ is a valid $\QLQSR$-formula.
\end{enumerate}
In view of (a)--(c) above, to prove that our problem is in
\textsf{NP}, it is enough to construct in nondeterministic polynomial
time a $\QLQSR$-interpretation and show that we can check in polynomial
time that it actually satisfies $\psi$.

Let $\model = (D,M)$ be a $\QLQSR$-model for $\psi$ and let $\model^*
= (D^*,M^*)$ be the relativized interpretation of $\model$ with respect
to a domain $D^{*}$ satisfying (\ref{DStar}), hence such that
$|D^{*}| = \mathcal{O}(|\psi|^{h+1})$, since $\psi$ is a
$(\QLQSR)^h$-formula (cf.\ Theorem \ref{correctness} and the
construction described in Sections \ref{ssseUniv} and \ref{machinery}).

In view of the remarks just before Definition~\ref{def:hlang}, to
complete our proof it is enough to check that
\begin{itemize}
\item $M^* X^2 \subseteq \pow_{< h}(D^*)$, for any free variable $X^2$
of sort 2 in $\psi$ (which entails that $|M^* X^2| =
\mathcal{O}(|D^{*}|^{h})$),

\item $M^* X^3 \subseteq \pow_{< h}(\pow_{< h}(D^*))$, for any free
variable $X^3$ of sort 3 in $\psi$ (which entails that $|M^* X^3| =
\mathcal{O}(|D^{*}|^{h^{2}})$), and

\item $\model^* \models \psi$ can be verified in deterministic
polynomial time.
\end{itemize}

To prove that $M^* X^2 \subseteq \pow_{< h}(D^*)$, for any free
variable $X^2$ in $\psi$, we reason as follows.  Let $X^2$ be a
variable of sort 2 occurring free in $\psi$.  From Definition
\ref{relintrp}, we recall that
\begin{eqnarray}
M^{*}X^2  &=&
    ((MX^2 \cap \pow(D^{*})) \setminus \{M^{*}X^1: X^1 \in ({\cal V}'_{1}\cup {\cal V}_1^F)\})
    \notag\\
     & & \qquad\cup \{M^{*}X^1: X^1 \in ({\cal V}'_{1} \cup {\cal
     V}_1^F),~MX^1 \in MX^2\}\, . \label{MStarXTwo}
\end{eqnarray}

Observe that
\begin{equation}
    MX^2 \subseteq \pow_{< h}(D). \label{MXTwo}
\end{equation}
Indeed, if the variable $X^2$ coincides with $X_{< h}^2$, then
(\ref{MXTwo}) follows from the fact that $\psi$ contains the conjunct
$\xi_{< h}^2$.  On the other hand, if $X^2$ is distinct from $X_{<
h}^2$, then $\psi$ contains either the conjunct $(\forall Z^1)(Z^1 \in X^2
\rightarrow Z^1 \in X_{< h}^2)$ or the conjunct $X^2 \in X_{< h}^3$. In the first case,
$(\forall Z^1)(Z^1 \in X^2
\rightarrow Z^1 \in X_{< h}^2)$ together with the conjunct
$\xi_{< h}^2$, implies again (\ref{MXTwo}).
From (\ref{MStarXTwo}) and (\ref{MXTwo}), we get $M^* X^2 \subseteq
\pow_{< h}(D^*)$. The other case is handled in a similar way.

Checking that $M^* X^3 \subseteq \pow_{< h}(\pow_{< h}(D^*))$, for any
free variable $X^3$ of sort 3 in $\psi$, can be carried out much as
was done for free variables of sort 2.

From what we have shown so far, it follows that
in nondeterministic polynomial time one can construct
\begin{itemize}
    \item the $(\QLQSR)^{h}$-formula $\psi$, as a result of
    applications of rules R1--R6 to the initial set $\Hf$
    (corresponding to the input formula $\varphi$) until saturation is
    reached,

    \item the $\QLQSR$-interpretation $\model^* = (D^*,M^*)$ (of
    $\psi$).
\end{itemize}
By the soundness of rules R1--R6, it follows that the $\QLQSR$-formula
$\psi \rightarrow \varphi$ is valid.  Thus, we obtain a succint
certificate of the satisfiability of $\varphi$ if we show that it is
possible to check in polynomial time that $\model^* \models \psi$
holds.  This is equivalent to show that we can check in polynomial
time that $\model^* \models \xi$, for every conjunct $\xi$ of $\psi$.
We distinguish the following cases.
\begin{description}
\item [$\xi$ is a quantifier-free atomic formula:] Since all variables
in $\xi$ are interpreted by $\model^*$ with sets of polynomial size,
the task of checking memberships and equalities among such sets can be
performed in polynomial time.

\begin{sloppypar}
\item [$\xi$ is a purely universal formula of level 1 $(\forall
z_1)\ldots(\forall z_n)\varphi_0$, with $n \leq h$:] We have that
$\model^* \models (\forall z_1)\ldots(\forall z_n)\varphi_0$ if and
only if $\model^*[z_1/u_1,\ldots,z_n/u_n] \models \varphi_0$, for
every $u_1,\ldots,u_n \in D^*$.  From the previous case, for any
$u_1,\ldots,u_n \in D^*$, one can compute in polynomial time whether
$\model^*[z_1/u_1,\ldots,z_n/u_n] \models \varphi_0$.  Since the
collection of such $n$-tuples $u_1,\ldots,u_n \in D^*$ has polynomial
size in $|\varphi|$, it turns out that one can check that $\model^*
\models (\forall z_1)\ldots(\forall z_n)\varphi_0$ in polynomial time.
\end{sloppypar}

\item [$\xi$ is a purely universal formula of level 2:]  If
\[
\xi \equiv \xi_{< h}^2 \equiv (\forall Z^1)\Big(Z^1 \in X_{< h}^2
\rightarrow \Big((\forall z_1)\ldots (\forall
      z_{h})\Big(\bigwedge_{i=1}^{h} z_i \in Z^1 \rightarrow
      \Big(\bigvee_{\substack{i,j=1\\  i<j}}^{h} z_i = z_j\Big)\Big)\Big)\Big) \, ,
\]
in order to verify that $\model^* \models \xi$, it is enough to check
that $\model^* X_{< h}^2 \subseteq \pow_{< h}(D^*)$, which can be
clearly done in polynomial time.

If $\xi \equiv (\forall Z^1)(Z^1 \in X^2 \rightarrow
Z^1 \in X_{< h}^2)$, with $X^2$ a free variable of sort 2, then in
order to verify that $\model^* \models \xi$ it is enough to check
whether $\model^* X^2 \subseteq \model^* X_{< h}^2$, which again can
be done in polynomial time.

Finally, if $\xi \equiv (\forall Z_1^1)\ldots (\forall Z_m^1)((Z_1^1
\in X_{< h}^2 \wedge \ldots \wedge Z_m^1 \in X_{< h}^2) \rightarrow
\varphi_1)$ where $m \leq h$ and $\varphi_1$ is a propositional
combination of quantifier-free atomic formulae and of purely universal formulae of level 1 of the form $(\forall z_1)\ldots(\forall
z_n)\varphi_0$, with $n \leq h$ (cf.\
Definition~\ref{def:hlang}(\ref{levelTwo})), then $\model^* \models
\xi$ if and only if $\model^*[Z_1^1/U_1^1,\ldots,Z_m^1/U_m^1] \models
\varphi_1$, for every $U_1^1, \ldots,U_m^1 \in M^*X_{< h}^2$.  Again,
the latter task can be accomplished in polynomial time, since, in view
of the previous two cases $\model^*[Z_1^1/U_1^1,\ldots,Z_m^1/U_m^1]
\models \varphi_1$ can be checked in polynomial time, for each
$m$-tuple $U_1^1, \ldots,U_m^1 \in M^*X_{< h}^2$, and the number of
such $m$-tuples is polynomial.

\item [$\xi$ is a purely universal formula of level 3:]
This case can be handled much along the same lines of the previous
case.
\end{description}
Summing up, we have shown that the satisfiability problem for
$(\QLQSR)^h$-formulae is in \textsf{NP}.  This, together with its
\textsf{NP}-hardness, which was shown before, implies the
\textsf{NP}-completeness of our problem.
\end{proof}

In the next section we show how the fragment $\QLQSR$ can be used
to formalize the modal logic $\Kqc$.

\subsection{Applying $\QLQSR$ to modal logic}
The \emph{modal language} $\mathsf{L}_{M}$ is based on a countably
infinite set of propositional letters ${\mathcal P} =
\{p_1,p_2,\ldots\}$, the classical propositional connectives `$\neg$',
`$\And$' , and `$\Or$', the modal operators `$\square$', `$\lozenge$'
(and the parentheses).  $\mathsf{L}_{M}$ is the smallest set such that
${\mathcal P} \subseteq \mathsf{L}_{M}$, and such that if
$\varphi,\psi \in \mathsf{L}_{M}$, then $\neg \varphi$, $\varphi \And
\psi$, $\varphi \Or \psi$, $\square \varphi$, $\lozenge\varphi\in
\mathsf{L}_{M}$.  Lower case letters like $p$ denote elements of
${\mathcal P}$ and Greek letters like $\varphi$ and $\psi$ represent
formulae of $\mathsf{L}_{M}$.  Given a formula $\varphi$ of
$\mathsf{L}_{M}$, we indicate with $\SubF(\varphi)$ the collection of
the subformulae of $\varphi$.  The \emph{modal depth} of a formula
$\varphi$ is the maximum nesting depth of modalities occurring in
$\varphi$. In the rest of the paper we also make use of the propositional connective `$\rightarrow$' defined in terms of
`$\neg$' and `$\Or$' as: $\varphi \rightarrow \psi \equiv \neg \varphi \Or \psi$.

A \emph{normal modal logic} $\mathsf{N}$ is any subset of $\mathsf{L}_{M}$
which contains all the tautologies and the axiom
$$
\K: \,\, (p_1 \rightarrow p_2) \rightarrow (\square p_1 \rightarrow \square p_2) \, ,
$$
and which is closed with respect to the following rules:
\begin{description}
\item [(Modus Ponens):] if $\varphi, \varphi \rightarrow \psi \in \mathsf{N}$, then $\psi \in \mathsf{N}$,
\item [(Necessitation):] if $\varphi \in \mathsf{N}$, then $\square \varphi \in \mathsf{N}$,
\item [(Substitution):] if $\varphi \in \mathsf{N}$, then $\mathsf{s}\, \varphi  \in \mathsf{N}$,
\end{description}
where $\varphi,\psi\in \mathsf{L}_{M}$, and the formula $\mathsf{s}\, \varphi$ is the result of uniformly substituting in $\varphi$ propositional letters with formulae (the reader may consult a text on modal logic like \cite{ModLog01} for more details).

A \emph{Kripke frame} is a pair $\langle W,R\rangle$ such that $W$ is
a nonempty set of possible worlds and $R$ is a binary relation on $W$
called \emph{accessibility relation}.  If $R(w,u)$ holds, we say that
the world $u$ is accessible from the world $w$.  A \emph{Kripke model}
is a triple $\langle W,R,h\rangle$, where $\langle W,R\rangle$ is a
Kripke frame and $h$ is a function mapping propositional letters into
subsets of $W$.  Thus, $h(p)$ is the set of all the worlds in which
$p$ is true.

Let $\Kr = \langle W,R,h\rangle$ be a Kripke model and let $w$ be a
world in $\Kr$. Then, for every $p \in \mathcal{P}$ and for every
$\varphi,\psi \in {\mathsf{L}_{M}}$, the satisfaction relation
$\models$ is defined as follows:
\begin{itemize}
\item $\Kr, w \models p$ ~iff~ $w \in h(p)$;

\item $\Kr, w \models \varphi \Or \psi$ ~iff~ $\Kr, w \models \varphi$ or $\Kr, w \models \psi$;

\item $\Kr, w \models \varphi \And \psi$ ~iff~ $\Kr, w \models \varphi$ and $\Kr, w \models \psi$;

\item $\Kr, w \models \neg \varphi$ ~iff~ $\Kr, w \not\models \varphi$;

\item $\Kr, w \models \square\varphi$ ~iff~ $\Kr, w' \models\varphi$,
for every $w'\in W$ such that $(w,w')\in R$;

\item $\Kr, w \models \lozenge\varphi$ ~iff~ there is a $w'\in W$ such that $(w,w')\in R$ and $\Kr, w' \models\varphi$.
\end{itemize}
A formula $\varphi$ is said to be \emph{satisfied} at $w$ in $\Kr$ if
$\Kr, w \models \varphi$; $\varphi$ is said to be \emph{valid} in
$\Kr$ (and we write $\Kr \models \varphi$), if $\Kr,w \models
\varphi$, for every $w \in W$.

The smallest normal modal logic is $\K$, which contains only the modal
axiom $\K$ and whose accessibility relation $R$ can be any binary
relation.  The other normal modal logics admit together with $\K$
other modal axioms drawn from the
ones in Table \ref{tab:modalax}.

Translation of a normal modal logic into the $\QLQSR$ language is
based on the semantics of propositional and modal operators. For any normal modal logic, the formalization of
the semantics of modal operators depends on the axioms that characterize the logic.

In the case of the logic $\Kqc$, whose decision problem has been shown
to be ${\mathsf{NP}}$-complete in \cite{Lad77}, the modal formulae
$\square\varphi$ and $\lozenge\varphi$ can be expressed in the
$\QLQSR$ language and thus the logic $\Kqc$ can be entirely translated
into the $\QLQSR$ fragment.  This is shown in what follows.

\begin{table}[tb]
\begin{center}
\begin{tabular}{|l|l|l|}
  \hline
Axiom & Schema & Condition on $R$ (see Table \ref{tab:accesrel}) \\\hline
$\T$  & $\square p \rightarrow p$ & Reflexive \\
$\Ac$ & $\lozenge p \rightarrow \square\lozenge p$ & Euclidean \\
$\B$  & $p \rightarrow \square\lozenge p$ & Symmetric \\
$\Aq$ & $\square p \rightarrow \square \square p$ & Transitive \\
$\D$  & $\square p \rightarrow \lozenge p$ & Serial: $(\forall w)(\exists u)R(w,u)$\\
\hline
\end{tabular}
\caption{\label{tab:modalax} Axioms of normal modal logics}
\end{center}
\end{table}

\subsubsection{The logic $\Kqc$}
The normal modal logic $\Kqc$ is obtained from the logic $\K$ by
adding to $\K$ the axioms
$\Aq$ and $\Ac$ listed in Table \ref{tab:modalax}.  Semantics of the
modal operators $\square$ and $\lozenge$ for the logic $\Kqc$ can be
described as follows.  Given a formula $\varphi$ of $\Kqc$ and a
Kripke model $\Kr = \langle W,R,h\rangle$, we put:
\begin{itemize}
\item $\Kr\models \square \varphi$ ~iff~ $\Kr,v \models \varphi$, for every $v \in W$ s.t. there is a $w' \in W$ with $(w',v)\in R$,
\item $\Kr\models \lozenge \varphi$ ~iff~ $\Kr,v \models \varphi$, for some $v \in W$ s.t. there is a $w' \in W$ with $(w',v)\in R$.
\end{itemize}
This formulation allows one to express a formula $\varphi$ of $\Kqc$
into the $\QLQSR$ fragment.  In order to simplify the definition of
the translation function $\TKqc$ introduced below, we give the notion
of the ``empty formula'', to be denoted by $\Lambda$, and which will
not be interpreted in any particular way.  The only requirement on
$\Lambda$ needed for the definitions to be given below is that
$\Lambda \And \psi$ and $\psi \And \Lambda$ must be regarded as
syntactic variations of $\psi$, for any $\QLQSR$-formula $\psi$.

Intuitively, the translation function $\TKqc$ associates to each formula $\varphi$
of $\Kqc$ a $\QLQSR$-formula defining a variable $X_{\varphi}$ of sort 1, which denotes
the subset $W_{\varphi}$ of $W$ such that $\Kr, w \models \varphi$ if and only if $w \in W_{\varphi}$, for every Kripke model
$\Kr = \langle W,R,h\rangle$. We proceed as follows.

For every propositional letter $p$, let $\TKqcuno(p) = X_{p}^{1}$,
with $X_{p}^{1}\in {\cal V}_1$, and let $\TKqcdue : \Kqc \rightarrow
\QLQSR$ be the function defined recursively as follows:

\begin{itemize}
\item $\TKqcdue(p) = \Lambda$,

\item $\TKqcdue(\neg \varphi) = (\forall z)(z \in X_{\neg \varphi}^{1}
\leftrightarrow \neg(z \in X_{\varphi}^{1})) \And \TKqcdue(\varphi)$,

\item $\TKqcdue(\varphi_1 \And \varphi_2) = (\forall z)(z \in
X_{\varphi_1 \And \varphi_2}^{1} \leftrightarrow (z \in
X_{\varphi_1}^{1} \And z \in X_{\varphi_2}^{1})) \And
\TKqcdue(\varphi_1) \And \TKqcdue(\varphi_2)$,

\item $\TKqcdue(\varphi_1 \Or \varphi_2) = (\forall z)(z \in
X_{\varphi_1 \Or \varphi_2}^{1} \leftrightarrow (z \in
X_{\varphi_1}^{1} \Or z \in X_{\varphi_2}^{1})) \And
\TKqcdue(\varphi_1) \And \TKqcdue(\varphi_2)$,

\item $\TKqcdue(\square\psi) = (\forall z_1)((\neg(\forall
z_2)\neg(\langle z_2,z_1\rangle \in X_R^{3})) \rightarrow z_1 \in
X_{\psi}^{1})\rightarrow (\forall z)(z \in X_{\square \psi}^{1})$

\hfill $\And \neg (\forall z_1)\neg ((\neg(\forall z_2)\neg(\langle
z_2,z_1 \rangle \in X_R^{3})) \And \neg(z_1 \in X_{\psi}^{1}))
\rightarrow (\forall z)\neg (z \in X_{\square \psi}^{1}) \And
\TKqcdue(\psi)$,

\item $\TKqcdue(\lozenge\psi) = \neg(\forall z_1)\neg ((\neg(\forall
z_2)\neg(\langle z_2,z_1\rangle \in X_R^{3})) \And z_1 \in
X_{\psi}^{1})\rightarrow (\forall z)(z \in X_{\lozenge \psi}^{1})$

\hfill $\And (\forall z_1)(((\forall z_2)\neg (\langle z_2,z_1\rangle
\in X_R^{3})) \Or \neg(z_1 \in X_{\psi}^{1}))) \rightarrow (\forall
z)\neg (z \in X_{\lozenge \psi}^{1}) \And \TKqcdue(\psi)$,
\end{itemize}
where $\Lambda$ is the empty formula, $X_{\neg \varphi}^{1}, X_{\varphi}^{1}, X_{\varphi_1 \And \varphi_2}^{1}, X_{\varphi_1 \Or \varphi_2}^{1}, X_{\varphi_1}^{1}, X_{\varphi_2}^{1} \in {\cal V}_1$, and $X_R^3\in {\cal V}_3$.

Finally, for every $\varphi$ in $\Kqc$, if $\varphi$ is a
propositional letter in $\mathcal{P}$ we put $\TKqc(\varphi) =
\TKqcuno(\varphi)$, otherwise $\TKqc(\varphi) = \TKqcdue(\varphi)$.
Next, by means of the following formulae, we characterize a variable
$X_R^3$ of sort 3, intended to denote the accessibility relation $R$
of the logic $\Kqc$:
\begin{itemize}
\item $\chi_1 = (\forall z_1)(\forall z_2)(\langle z_1, z_2\rangle \in X_{< 3}^3)$,

\item $\chi_2 = (\forall Z^{2})((Z^{2} \in X_{< 3}^3) \rightarrow ((Z^{2} \in X_R^{3} \leftrightarrow \neg (\forall
    z_1)(\forall z_2)\neg(\langle z_1,z_2\rangle = Z^{2}))))$,

\item $\chi_3 = (\forall z_1,z_2,z_3)((\langle
    z_1,z_2\rangle \in X_R^3 \And \langle z_2,z_3\rangle \in
    X_R^3)\rightarrow \langle z_1,z_3\rangle \in X_R^3)$,

\item $\chi_4 = (\forall z_1,z_2,z_3)((\langle
    z_1,z_2\rangle \in X_R^3 \And \langle z_1,z_3\rangle \in
    X_R^3)\rightarrow \langle z_2,z_3\rangle \in X_R^3)$,

\item $\psi_1^2 = (\forall Z^2)(Z^2 \in X_R^3 \rightarrow Z^2 \in X_{<
3}^3)$.
\end{itemize}

Correctness of the translation is stated by the following lemma.
\begin{lemma}\label{leK45}
For every formula $\varphi$ of the logic $\TKqc$, $\varphi$ is
satisfiable in a model $\Kr = \langle W,R,h\rangle$ if and only if
there is a
$\QLQSR$-interpretation satisfying $x \in X_{\varphi}^1$.
\end{lemma}
\begin{proof}
Let $\bar{w}$ be a world in $W$.  We construct a
$\QLQSR$-interpretation $\model=(W,M)$ as follows:
\begin{itemize}
\item $Mx = \bar{w}$,

\item $M X_{p}^{1} = h(p)$, where $p$ is a propositional letter and
$X_{p}^{1} = \TKqc(p)$,

\item $M\TKqc(\psi) = \true$, for every subformula $\psi$ of
$\varphi$,
distinct from a propositional letter.
\end{itemize}
To prove the lemma, it would be enough to show that $\Kr,\bar{w}
\models \varphi$ ~iff~ $M \models x \in X_{\varphi}^{1}$.  However, it
is more convenient to prove the following more general property:\\
\indent \emph{Given a $w \in W$ and a $y \in {\cal V}_0$ such that $My =
w$, we have
$$
\Kr,w \models \varphi \mbox{ ~~iff~~ } M \models y \in
X_{\varphi}^1.
$$}
We proceed by structural induction on $\varphi$ by considering for
simplicity only the relevant cases in which $\varphi = \square \psi$
and $\varphi = \lozenge \psi$.
\begin{itemize}
\item Let $\varphi = \square \psi$ and assume that $\Kr,w \models
\square \psi$.  Let $v$ be a world of $W$ such that $\langle
u,v\rangle \in R$ for some $u\in W$, and let $x_1,x_2 \in {\cal V}_0$
be such that $v = Mx_1$ and $u = Mx_2$.  We have that $\Kr,v \models
\psi$ and, by inductive hypothesis, $M \models x_1 \in X_{\psi}^{1}$.
Since $M \models \TKqc(\square \psi)$, then $M \models (\forall
z_1)((\neg(\forall z_2)\neg (\langle z_2,z_1\rangle\in X_R^{3}))
\rightarrow z_1 \in X_{\psi}^{1}) \rightarrow (\forall z)(z \in
X_{\square \psi}^{1})$.  Hence $M[z_1/v,z_2/u,z/w]\models (\langle
z_2,z_1\rangle\in X_R^{3} \rightarrow z_1 \in X_{\psi}^{1})
\rightarrow z \in X_{\square \psi}^{1}$ and thus $M \models (\langle
x_2,x_1\rangle\in X_R^{3} \rightarrow x_1 \in X_{\psi}^{1})
\rightarrow y \in X_{\square \psi}^{1}$.  Since $M \models \langle
x_2,x_1\rangle\in X_R^{3} \rightarrow x_1 \in X_{\psi}^{1}$, by modus
ponens we have the thesis.  The thesis follows also in the case in
which there is no $u$ such that $\langle u,v\rangle\in X_R^{3}$.  In
fact, in that case $M \models \langle x_2,x_1\rangle\in X_R^{3}
\rightarrow x_1 \in X_{\psi}^{1}$ holds for any $x_2 \in {\cal V}_0$.

Consider next the case in which $\Kr,w \not\models \square \psi$.
Then, there must be a $v \in W$ such that $\langle u,v \rangle\in
X_R^{3}$, for some $u \in W$, and $\Kr,v \not\models \psi$.  Let
$x_1,x_2 \in {\cal V}_0$ be such that $Mx_1 = v$ and $Mx_2 = u$.
Then, by inductive hypothesis, $M \not\models x_1 \in X_{\psi}^{1}$.

By definition of $M$, we have $M \models \neg (\forall z_1)\neg
((\neg(\forall z_2)\neg (\langle z_2,z_1\rangle\in X_R^{3})) \And
\neg(z_1 \in X_{\psi}^{1})) \rightarrow (\forall z)\neg (z \in
X_{\square \psi}^{1})$.  By the above instantiations and by the
hypotheses, we have that $M \models ((\langle x_2,x_1\rangle\in X_R^{3})
\And \neg(x_1 \in X_{\psi}^{1}))\rightarrow \neg(y \in
X_{\square\psi}^{1})$ and $M \models (\langle x_2,x_1\rangle\in
X_R^{3}) \And \neg(x_1 \in X_{\psi}^{1})$.  Thus, by modus ponens, we
obtain the thesis.

\item Let $\varphi = \lozenge \psi$ and assume that $\Kr,w \models
\lozenge \psi$.  Then there are $u,v \in W$ such that $\langle u,v
\rangle\in R$ and $\Kr,v \models \psi$.  Let $x_1,x_2 \in {\cal V}_0$
be such that $Mx_1 = v$ and $Mx_2 = u$.  Then, by inductive
hypothesis, $M \models x_1 \in X_{\psi}^{1}$.  Since $M \models
\TKqc(\lozenge \psi)$, it follows that $M \models \neg (\forall
z_1)\neg ((\neg(\forall z_2)\neg(\langle z_2,z_1\rangle \in X_R^{3}))
\And z_1 \in X_{\psi}^{1})\rightarrow (\forall z)(z \in X_{\lozenge
\psi}^{1})$.  By the hypotheses and the variable instantiations above
it follows that $M \models ((\langle x_2,x_1\rangle \in X_R^{3}) \And
x_1 \in X_{\psi}^{1}) \rightarrow y \in X_{\lozenge\psi}^{1}$ and $M
\models (\langle x_2,x_1\rangle \in X_R^{3}) \And x_1 \in X_{\psi}^{1}$.
Finally, by an application of modus ponens the thesis follows.

On the other hand, if $\Kr,w \not\models \lozenge \psi$, then for
every $v \in W$, either there is no $u \in W$ such that $\langle u,v
\rangle\in R$, or $\Kr,v \not\models \psi$.  Let $x_1,x_2 \in {\cal
V}_0$ be such that $Mx_1 = v$ and $Mx_2 = u$.  If $\Kr,v \not\models
\psi$, by inductive hypothesis, we have that $M \not\models y \in
X_{\psi}^{1}$.

Since $M \models (\forall z_1) (((\forall z_2)\neg (\langle
z_2,z_1\rangle \in X_R^{3})) \Or \neg(z_1 \in X_{\psi}^{1}))\rightarrow
(\forall z)\neg(z \in X_{\lozenge \psi}^{1})$, by the hypotheses and
by the variable instantiations above we get $M \models (\neg (\langle
x_2,x_1\rangle \in X_R^{3}) \Or \neg( x_1 \in X_{\psi}^{1})) \rightarrow
\neg (y \in X_{\lozenge\psi}^{1})$ and $M \models (\neg(\langle
x_2,x_1\rangle \in X_R^{3}) \Or \neg( x_1 \in X_{\psi}^{1}))$.  Finally,
by modus ponens we infer the thesis.
\end{itemize}
\end{proof}
It can be easily verified that $\TKqc(\varphi)$ is polynomial in the
size of $\varphi$ and that its satisfiability can be checked in
nondeterministic polynomial time since the formula
$$
\xi_W^1 \And \xi_{< 3}^2 \And \xi_{< 3}^3 \And \psi_1^2 \And (\chi_1
\And \chi_2 \And \chi_3 \And \chi_4 \And \TKqc(\varphi))
$$
belongs to $(\QLQSR)^3$.\footnote{$\xi_W^1$ is intended to
characterize a nonempty set of possible worlds.} Thus, the decision
algorithm for $\QLQSR$ we have presented and the translation function
described above yield a nondeterministic polynomial decision procedure
for testing the satisfiability of any formula $\varphi$ of $\Kqc$.

\section{Conclusions and future work}
We have presented a decidability result for the satisfiability problem
for the fragment $\QLQSR$ of multi-sorted stratified syllogistic
embodying variables of four sorts and a restricted form of
quantification.  As the semantics of the modal formulae $\square
\varphi$ and $\lozenge \varphi$ in the modal logic $\Kqc$
can be easily formalized in a fragment of $\QLQSR$, admitting a
nondeterministic polynomial decision procedure, we obtained an
alternative proof of the \textsf{NP}-completeness of $\Kqc$. The results
reported in the paper offer numerous hints of future work, some of which are discussed in what follows.

 Recently, we have analyzed several fragments of elementary set theory. It
will be interesting to ameliorate existing techniques to verify in a formal way the truth
of expressivity results that for the moment we have only conjectured. Moreover, we plan to
find complexity results for the fragments $\TLQSR$ (cfr. \cite{CanNic08}) and $\QLQSR$, and for some of their sublanguages
like, for instance, the sublanguages of $\QLQSR$ characterized by the fact that quantifier prefixes have
length bounded by a constant. According to the construction of Section \ref{ssseUniv} small models for
formulae of these sublanguages have a finite domain $D^*$ that is polynomial in the size of the formula. However, their
formulae are not subject to the syntactical constraints characterizing formulae of the $(\QLQSR)^h$ languages and allowing
the satisfiability problem for the $(\QLQSR)^h$ fragments to be \textsf{NP}-complete.

As we mentioned in the Introduction, stratified syllogistics have been studied less than one sorted multi-level ones.
Thus, a comparison of the results obtained in this paper with the results regarding one sorted multi-level
set theoretic decidability is in order.

Formalizations of modal logics in set theory have already been
provided within the framework of hyperset theory \cite{BaMo96} and of
weak set theories \cite{DMP95}, without the extensionality and
foundation axioms.

We intend to continue our study, started with \cite{CanNic08},
concerning the limits and possibilities of expressing modal, and
more generally, non-classical logics in the context of stratified
syllogistics.  Currently, in the case of modal logics characterized by
a liberal accessibility relation like $\K$, we are not able to
translate the modal formulae $\square \varphi$ and $\lozenge \varphi$
in $\QLQSR$. We plan to verify if $\QLQSR$ allows one to express modal logics with nesting of modal operators of bounded length.
We also intend to investigate extensions of $\QLQSR$ which
allow one to express suitably constrained occurrences of the
composition operator on binary relations and of the set-theoretic
operator of general union. We expect that these extensions will make it possible to express
all the normal modal logic systems and several multi-modal logics.
Finally, since within $\QLQSR$ we are able to express Boolean
operations on relations, we plan to investigate the possibility of
translating fragments of Boolean modal logic and expressive description logics
admitting boolean constructors over roles.

\bibliographystyle{plain}

\end{document}